\documentclass[a4paper]{cas-sc} 
\ExplSyntaxOn 
\bool_set_true:N \g_stm_nologo_bool
\ExplSyntaxOff
\usepackage{calc,amsfonts,amsmath,amssymb,amsthm,graphicx}
\usepackage{mathtools}
\usepackage[utf8]{inputenc}
\usepackage{CJKutf8}
\usepackage[shortlabels,inline]{enumitem}
\usepackage{multicol}
\usepackage[nameinlink]{cleveref}
\usepackage{comment}
\usepackage{epigraph}
\usepackage[numbers]{natbib}

\crefformat{page}{#2page~#1#3}%
\Crefformat{page}{#2Page~#1#3}%
\crefformat{equation}{#2(#1)#3}%
\Crefformat{equation}{#2(#1)#3}%
\crefformat{figure}{#2Figure~#1#3}%
\Crefformat{figure}{#2Figure~#1#3}%
\crefformat{theorem}{#2Theorem~#1#3}%
\Crefformat{theorem}{#2Theorem~#1#3}%
\crefformat{lemma}{#2Lemma~#1#3}%
\Crefformat{lemma}{#2Lemma~#1#3}%
\crefformat{corollary}{#2Corollary~#1#3}%
\Crefformat{corollary}{#2Corollary~#1#3}%
\crefformat{conjecture}{#2Conjecture~#1#3}%
\Crefformat{conjecture}{#2Conjecture~#1#3}%
\Crefformat{problem}{#2Problem~#1#3}%
\crefformat{problem}{#2Problem~#1#3}%
\Crefformat{open}{#2Open problem~#1#3}%
\crefformat{open}{#2Open problem~#1#3}%
\Crefformat{definition}{#2Definition~#1#3}%
\crefformat{definition}{#2Definition~#1#3}%
\crefformat{example}{#2Example~#1#3}%
\Crefformat{example}{#2Example~#1#3}%
\crefformat{section}{#2Section~#1#3}
\Crefformat{section}{#2Section~#1#3}
\crefformat{chapter}{#2Chapter~#1#3}
\Crefformat{chapter}{#2Chapter~#1#3}
\crefformat{chapter*}{#2Chapter~#1#3}
\Crefformat{chapter*}{#2Chapter~#1#3}
\crefformat{part}{#2Part~#1#3}
\Crefformat{part}{#2Part~#1#3}
\crefformat{enumi}{#2(#1)#3}
\Crefformat{enumi}{#2(#1)#3}
\crefformat{algorithm}{#2algorithm~#1#3}%
\Crefformat{algorithm}{#2Algorithm~#1#3}%

\usepackage{mathrsfs}

\usepackage[all]{xy}

\theoremstyle{plain}
\newtheorem{theorem}{Theorem}[section]
\newenvironment{xrefthm}[1]{%
  \def\thexref{\ref{#1}}
  \begin{thexrefthm}
}{%
  \end{thexrefthm}
}
\newtheorem{corollary}{Corollary}[section]

\newtheorem*{thexrefthm}{Theorem \thexref}
\newtheorem*{thexrefcor}{Corollary \thexref}
\newtheorem{lemma}[theorem]{Lemma}
\newtheorem{proposition}[theorem]{Proposition}

\theoremstyle{definition}
\newtheorem{definition}[theorem]{Definition}
\newtheorem{conjecture}[theorem]{Conjecture}
\newtheorem{remark}[theorem]{Remark}
\newtheorem{problem}[theorem]{Problem}

\newtheorem{example}[theorem]{Example}

\renewcommand{\phi}{\varphi}
\renewcommand{\mid}{~:~}
\newcommand{\Cc}{\mathscr{C}}
\newcommand{\Dd}{\mathscr{D}}
\newcommand{\Gg}{\mathscr{G}}
\newcommand{\Ww}{\mathscr{W}}
\newcommand{\Ff}{\mathscr{F}}

\DeclareMathOperator{\rw}{\mathrm{rw}}
\DeclareMathOperator{\lrw}{\mathrm{lrw}}
\DeclareMathOperator{\Class}{\mathrm{Class}}
\DeclareMathOperator{\IC}{\mathrm{ICol}}
\DeclareMathOperator{\NC}{\mathrm{NCol}}

\newcommand{\limp}{\mathbin{\rightarrow}}
\newcommand{\N}{\mathbb{N}}

\newcommand{\str}[1]{\mathbf #1}

\newcommand{\interp}[1]{\mathsf{#1}}

\begin{document}
\let\WriteBookmarks\relax
\def\floatpagepagefraction{1}
\def\textpagefraction{.001}

\title[mode=title]{Classes of graphs with low complexity:\\
 the case of classes with bounded linear rankwidth}\tnotemark[1]
 \tnotetext[1]{Submitted to the special issue of the {\em European Journal of Combinatorics} celebrating Xuding Zhu's sixtieth birthday.}
\shorttitle{Classes of graphs with low complexity}
\shortauthors{J. Ne{\v{s}}et{\v{r}}il et~al.}
\author{Jaroslav Ne{\v{s}}et{\v{r}}il}[orcid=0000-0002-5133-5586]\fnmark[1,2]
\address[1]{Institute for Theoretical Computer Science 
Charles University
Prague, Czech Republic}
\ead{nesetril@iuuk.mff.cuni.cz}

\author{Patrice {Ossona de Mendez}}[orcid=0000-0003-0724-3729]\fnmark[2]
\address[2]{Centre d'Analyse et de Math\'ematique Sociales (UMR 8557)
Centre National de la Recherche Scientifique,
Paris, France}
\ead{pom@ehess.fr}

\author{Roman Rabinovich}\fnmark[3]
\address[3]{Technical University Berlin,
Germany}
\ead{roman.rabinovich@tu-berlin.de}
\author{Sebastian Siebertz}[orcid=0000-0002-6347-1198]
\cormark[1]
\address[4]{University of Bremen,
 Germany}
\ead{siebertz@uni-bremen.de}
\cortext[cor1]{Corresponding author}
\fntext[fn1]{Supported by  CE-ITI P202/12/G061 of GA\v{C}R}
\fntext[fn2]{Supported by by the European 
Associated Laboratory (LEA STRUCO), and
by the
European Research Council (ERC) under the European Union's Horizon
2020 research and innovation programme (ERC Synergy Grant DYNASNET, grant agreement No 810115).}
\fntext[fn3]{Supported by Deutsche Forschungsgemeinschaft (DFG) --- ``Graph Classes of Bounded Shrubdepth'' Projekt number 420419861}

\begin{abstract}
  Classes with bounded rankwidth are MSO-transductions of trees and
  classes with bounded linear rankwidth are MSO-transductions of paths
  --
  a result that shows a strong link between the properties of these
  graph classes considered from the point of view of structural graph
  theory and from the point of view of finite model theory.
  We take both views on classes with bounded linear rankwidth and
  prove structural and model theoretic properties of these
  classes. The structural results we obtain are the following.
  1) The number of unlabeled graphs of order $n$ with linear
  rank-width at most~$r$ is at most
  $\bigl[(r/2)!\,2^{\binom{r}{2}}3^{r+2}\bigr]^n$.
  2) Graphs with linear rankwidth at most $r$ are linearly
  $\chi$-bounded. Actually, they have bounded $c$-chromatic number,
  meaning that they can be colored with $f(r)$ colors, each color
  inducing a cograph.
  3) To the contrary, based on a Ramsey-like argument, we prove for
  every proper hereditary family $\Ff$ of graphs (like
  cographs) that there is a class with bounded rankwidth that does not
  have the property that graphs in it can be colored by a bounded
  number of colors, each inducing a subgraph in $\Ff$.
  
  From the model theoretical side we obtain the following results:
  1) A direct short proof that graphs with linear rankwidth at most
  $r$ are first-order transductions of linear orders. This result
  could also be derived from Colcombet's theorem on first-order
  transduction of linear orders and the equivalence of linear
  rankwidth with linear cliquewidth.
  2) For a class $\mathscr C$ with bounded linear rankwidth the
  following conditions are equivalent: a) $\mathscr C$ is stable, b)
  $\mathscr C$ excludes some half-graph as a semi-induced subgraph, c)
  $\mathscr C$ is a first-order transduction of a class with bounded
  pathwidth.
  %
  These results open the perspective to study classes admitting low
  linear rankwidth covers.
  
  \mbox{ }

  \mbox{ }
  
  \mbox{ }
%
%
\end{abstract}

\begin{keywords}
rankwidth \sep linear rankwidth \sep cliquewidth \sep linear cliquewidth \sep linear NLC-width \sep pathwidth \sep coloring \sep c-coloring \sep cographs \sep $\chi$-bounded \sep low shrubdepth coloring \sep monadic stability \sep monadic dependence \sep first-order transduction \sep structurally bounded expansion 
\MSC[2010]{05C75 (Structural characterization of families of graphs), 05C15 (Coloring of graphs and hypergraphs), 	05C50 (Graphs and linear algebra), 03C13 (Finite structures), 03C45 (Classification theory, stability and related concepts)}
\end{keywords}
\maketitle


\setlength\epigraphwidth{.5\textwidth}
\epigraph{On devient jeune \`a soixante ans. \\
	Malheureusement, c'est trop tard.\\
	~\\	
	{\em You become young when you're sixty.\\
	Unfortunately, it's too late.}\\
	~\\
	\begin{CJK*}{UTF8}{gbsn}
	到60岁，我们才开始变得年轻。\\
	不幸的是，为时晚矣。
\end{CJK*}
	}{Pablo Picasso}
	
	\pagebreak

\tableofcontents

\section{Introduction}
\noindent A primary concern in many areas of mathematics is to classify structures (or classes of structures) according to their intrinsic complexity.  In this paper we consider three approaches and their interplay to the notion of structural complexity:
the model theoretic approach based on the standard dividing lines that are stability and dependence, the algebraic approach founding the notion of rankwidth and linear rankwidth, and a more classical graph theoretical approach based on colorings and decompositions of graphs.
%
%
%

A theory of sparse structures was initiated in~\cite{Sparsity}, which mainly fits to the classification of monotone classes. The theory has led to the nowhere dense/somewhere dense dichotomy that can be observed in several areas of graph theory, theoretical computer science, model theory, analysis, category theory and probability theory.  Motivated by the connection with model theory -- nowhere dense classes are monadically stable~\cite{adler2014interpreting} and even have low VC-density~\cite{pilipczuk2018number}~-- and by a possible extension of first-order model-checking algorithms for bounded expansion classes~\cite{DKT,DKT2} and for nowhere dense classes~\cite{Grohe2013}, these notions were extended to classes that are obtained as first-order transductions of sparse classes, the \emph{structurally sparse classes} \cite{SurveyND, SBE_drops}.  The central tool used in our approach is the transduction machinery, which establishes a fruitful bridge between graph theory and finite model theory.  Informally, a first-order transduction is a way to interpret a structure in another structure, where the new structure is defined by means of first-order formulas with set parameters.  Indeed, a standard approach of both model theory and computability theory is to determine the relative complexity of two structures by showing that the first interprets in the second, and is therefore not more complex than the second.  In this context, important classes of structures are the class of finite linear orders and the class of element to finite set membership graphs (powerset graphs), as they define the two most important model theoretical dividing lines: {\em stability}, which corresponds to the impossibility to interpret arbitrarily large linear orders, and {\em dependence} (or {\em NIP}, for ``Non-Independence Property''), which corresponds to the the impossibility to interpret arbitrarily large membership graphs. The versions of these properties where we allow set parameters are {\em monadic stability} and {\em monadic dependence}.

The use of first-order transductions naturally fits the study of hereditary classes. 
If we consider classes that are obtained as first-order transductions of other classes, the natural tractability limit is the realm of monadically NIP structures, as non monadically NIP classes allow to interpret the whole class of finite graphs. In this world, typical well behaved monadically NIP but monadically unstable classes of graphs are classes with bounded rankwidth (like cographs) and classes with bounded linear rankwidth (like half-graphs). This justifies a specific study of these classes, as well as the classes that admit finite $p$-covers with bounded rankwidth \cite{kwon17} or classes that admit finite $p$-covers with bounded linear rankwidth (like unit interval graphs), as they naturally extend structurally bounded expansion classes, which admit finite $p$-covers with bounded shrubdepth \cite{SBE_drops}.  However we do not know whether classes with such covers are monadically NIP.  The whole framework is schematically pictured on \Cref{fig:Universe}.

\begin{figure}
  \begin{center}
    \includegraphics[width=1\textwidth]{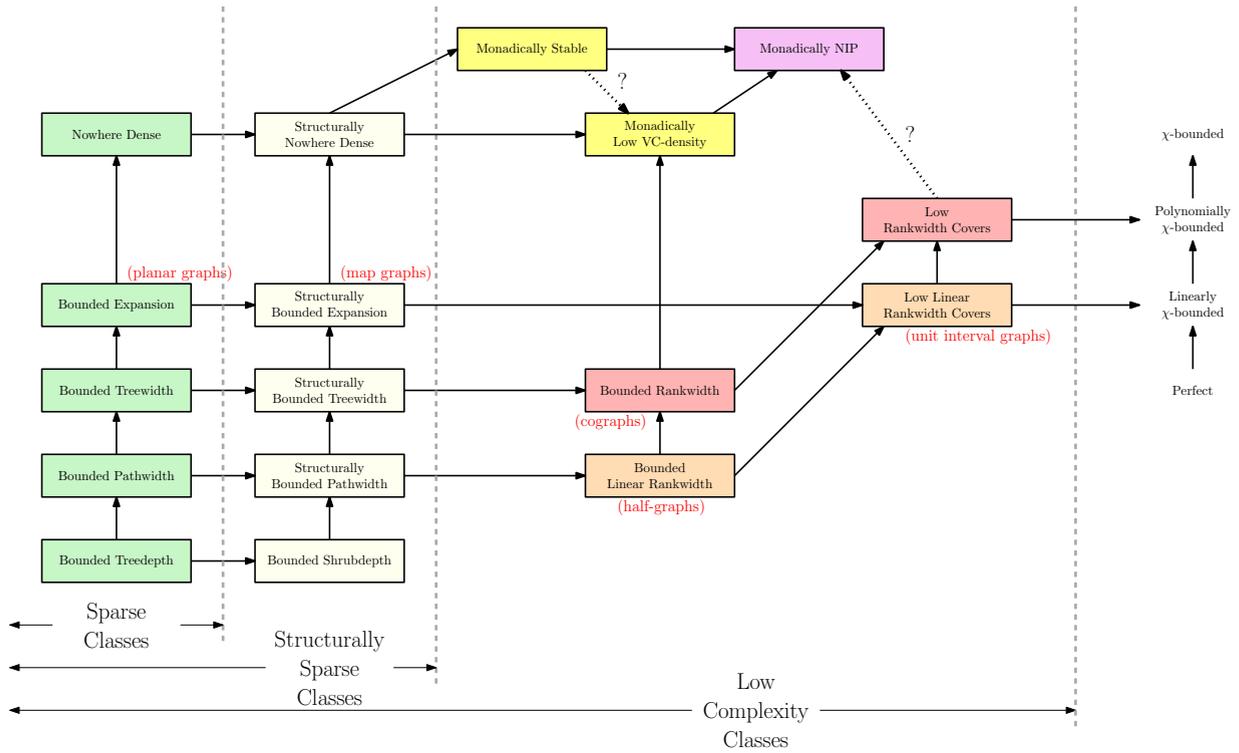}
  \end{center}
  \caption{Inclusion map of graph classes. Some examples of classes are given in brackets. }
  \label{fig:Universe} 
\end{figure}

This paper consists of two parts. The first part sets the scene and builds the framework that supports our study.  The second part roots our study in concrete problems.  In particular, we consider classes with bounded linear rankwidth and show how model theoretic and structural properties of classes with bounded linear rankwidth allow to prove new properties of these classes.  In particular we prove the following theorems (formal definitions will be given in \Cref{sec:prelims}).

\begin{xrefthm}{thm:stlrw}
  Let $\Cc$ be a class of graphs with bounded linear rankwidth. Then the following are equivalent:

  \begin{multicols}{2}
    \begin{enumerate}
    \item $\Cc$ is stable,
    \item $\Cc$ is monadically stable,
    \item $\Cc$ has $2$-covers with bounded shrubdepth,
    \item $\Cc$ is sparsifiable,
    \item $\Cc$ excludes some semi-induced half-graph,
    \item $\Cc$ is a first-order transduction of a class with bounded expansion (i.e.\ has structurally bounded expansion),
    \item $\Cc$ is a first-order transduction of a class with bounded pathwidth (i.e.\ has structurally bounded pathwidth).
    \end{enumerate}
  \end{multicols} 
\end{xrefthm} 

And we deduce 

\begin{xrefthm}{thm:lowlrwcover}
  Let $\Cc$ be a class with low linear rankwidth covers. Then the following are equivalent:
  \begin{enumerate}
  \item $\Cc$ is monadically stable,
  \item $\Cc$ is stable,
  \item $\Cc$ excludes a semi-induced half-graph,
  \item $\Cc$ has structurally bounded expansion.
  \end{enumerate} 
\end{xrefthm}

From the graph theoretic point of view, we briefly discuss how classes with bounded rankwidth differ from classes with bounded linear rankwidth and give some lower bounds for $\chi$-boundedness of graphs with bounded rankwidth and for graphs with bounded linear rankwidth.  Then we prove upper bounds for graphs with bounded linear rankwidth.

\begin{xrefthm}{thm:cog}
  Every graph $G$ with linear rankwidth at most $r$ can be colored with at most $3(r+2)!2^{\binom{r+1}{2}}$ colors such that each color induces a cograph with cotree height at most $r+2$.
  In particular, for every graph $G$ with linear rankwidth at most $r$ we have 
\[ \chi(G)\leq 3(r+2)!2^{\binom{r+1}{2}}\,\omega(G).  \] 
\end{xrefthm}

\Cref{thm:stlrw} and a weaker form of \Cref{thm:cog} (\Cref{thm:cog0}) are proved in \Cref{sec:NLC} by using the notion of linear NLC-width expression and Simon's factorization forest theorem.

The strong form of \Cref{thm:cog} is proved in \Cref{sec:lrw} by a fine analysis of linear rankwidth decompositions. Along the way we also obtain an upper bound for the number of graphs with linear rankwidth at most $r$.

\begin{xrefthm}{thm:numlrw}
  Unlabeled graphs with linear rankwidth at most $r$ can be encoded using at most
  $\binom{r}{2}+r\log_2 r+\log_2(3/e)r+O(\log_2 r)$ bits per vertex.
  Precisely, the number of unlabelled graphs of order $n$ with linear rankwidth at most~$r$ is at most $\left[(r+2)!\,2^{\binom{r}{2}}3^{r+2}\right]^{n}$.  
\end{xrefthm}

\section{Classes with low complexity}\label{sec:prelims}

\subsection{Structures and logic}  

\smallskip
\noindent A {\em{signature}} $\Sigma$ is a finite set of relation and
function symbols, each with a prescribed arity.  In this paper we
consider only signatures with relation symbols. A
\mbox{$\Sigma$-\em{structure}}~$\str A$ consists of a finite {\em
  universe} (or {\em domain}) $V(\str A)$ and interpretations of the
symbols in the signature: each relation symbol $R\in \Sigma$, say of
arity $k$, is interpreted as a $k$-ary relation
$R^{\str A}\subseteq V(\str A)^k$.  For a signature $\Sigma$, we
consider standard first-order logic over $\Sigma$.
If $\str A$ is a structure and $X\subseteq V(\str A)$ then we denote
by $\str A[X]$ the \emph{substructure} of $\str A$ induced by $X$.
The {\em{Gaifman graph}} of a structure $\str A$ is the graph with
vertex set $V(\str A)$ where two distinct elements $u,v\in \str A$ are
adjacent if and only if $u$ and $v$ appear together in some tuple in
some relation of~$\str A$.
For a formula $\phi(x_1,\dots,x_k)$ with $k$ free variables and a
structure~$\str A$, we define
\[
\phi(\str A)=\{(v_1,\dots,v_k)\in V(\str A)^k\mid \str A\models
\phi(v_1,\dots,v_k)\}.
\]

We usually write $\bar{x}$ for a tuple $(x_1,\dots,x_k)$ of variables
and leave it to the context to determine the length of the tuple. The
above equality then rewrites as
$\phi(\str A)=\{\bar{v}\in V(\str A)^{|\bar x|}\mid \str A\models
\phi(\bar{v})\}$.
Also, for a formula $\phi(\bar x,\bar y)$ and
$\bar{b}\in V(\str A)^{|\bar y|}$ we define
\[
\phi(\bar{b},\str A)=\{\bar{v}\in V(\str A)^{|\bar x|}\mid \str
A\models \phi(\bar v,\bar b)\}.
\]

A {\em monadic lift} $\Lambda$ of a $\Sigma$-structure $\mathbf A$ is
a $\Sigma^+$-structure $\Lambda(\mathbf A)$ such that $\Sigma^+$ is
the union of $\Sigma$ and a set of unary relation symbols and
$\mathbf A$ is the {\em shadow} of $\Lambda(\mathbf A)$, that is the
$\Sigma$-structure obtained from $\Lambda(\mathbf A)$ by
``forgetting'' all the relations in $\Sigma^+\setminus\Sigma$.

\subsection{Graphs, colored graphs and trees.}  

\smallskip
\noindent Graphs can be viewed as finite structures over the signature
consisting of a binary relation symbol~$E$, interpreted as the edge
relation, in the usual way.  For a finite label set $\Gamma$, by a
{\em{$\Gamma$-colored}} graph we mean a graph enriched by a unary
predicate~$U_\gamma$ for each $\gamma\in \Gamma$.  A rooted forest is
an acyclic graph $F$ together with a unary predicate $R\subseteq V(F)$
selecting one root in each connected component of~$F$. A tree is a
connected forest.  The {\em{depth}} of a node $x$ in a rooted
forest~$F$ is the number of vertices in the unique path between~$x$
and the root of the connected component of~$x$ in~$F$. In particular,
$x$ is a root of $F$ if and only if $F$ has depth $1$ in $F$.  The
depth of a forest is the largest depth of any of its nodes.  The
{\em{least common ancestor}} of nodes $x$ and $y$ in a rooted tree is
the common ancestor of~$x$ and $y$ that has the largest depth.

\subsection{Sparse graph classes}
\smallskip
\paragraph*{Treewidth, pathwidth and treedepth.} 
Treewidth is an important width parameter of graphs that was
introduced in~\cite{RS-GraphMinorsII-JAlg86} as part of the graph
minors project.  Pathwidth is a more restricted width measure that was
introduced in~\cite{RS-GraphMinorsI-JCTB83}.  The notion of treedepth
was introduced in~\cite{Taxi_tdepth}.

For our purposes it will be convenient to define these width measures
in terms of intersection graphs.  Let $S_1,\ldots, S_n$ be a family of
sets. The \emph{intersection graph} defined by this family is the
graph with vertex set $\{v_1,\ldots, v_n\}$ and edge set
$\{\{v_i,v_j\} : S_i\cap S_j\neq \emptyset\}$.

A {\em chordal graph} is the intersection graph of a family of
subtrees of a tree. An {\em interval graph} is the intersection graph
of a family of intervals. A {\em trivially perfect graph} is the
intersection graph of a family of nested intervals. Alternatively, a
trivially perfect graph is the comparability graph of a bounded depth
tree order.

The {\em treewidth} of a graph $G$ is one less than the minimum clique
number of a chordal supergraph of $G$, the {\em pathwidth} of a graph
$G$ is one less than the minimum clique number of an interval
supergraph of $G$, and the {\em treedepth} of a graph~$G$ is the
minimum clique number of a trivially perfect supergraph of $G$:
\begin{align*}
  \mathrm{tw}(G)&=\min\{\omega(H)-1\mid H\text{ chordal and }H\supseteq G\},\\
  \mathrm{pw}(G)&=\min\{\omega(H)-1\mid H\text{ interval graph and }H\supseteq G\},\\
  \mathrm{td}(G)&=\min\{\omega(H)\mid H\text{ trivially perfect and }H\supseteq G\}.\\	
\end{align*}

A class $\Cc$ of graphs has \emph{bounded treewidth, bounded
  pathwidth, or bounded treedepth}, respectively, if there is a bound
$k \in \mathbb N$ such that every graph in $\Cc$ has treewidth,
pathwidth, or treedepth, respectively, at most $k$.

\paragraph*{Classes with bounded expansion.}

A graph $H$ is a \emph{depth-$r$ topological minor} of a graph $G$ if
$G$ contains a subgraph isomorphic to a $\leq 2r$-subdivision of $H$.
A class $\Cc$ of graphs has \emph{bounded expansion} if there is a
function $f\colon \mathbb{N} \rightarrow \mathbb{N}$ such that
$\frac{\|H\|}{|H|}\leq f(r)$ for every $r \in \mathbb{N}$ and every
depth-$r$ topological minor $H$ of a graph from $\Cc$.  Examples of
classes with bounded expansion include the class of planar graphs, any
class of graphs with bounded maximum degree, or more generally, any
class of graphs that excludes a fixed topological minor.  We lift the
notion with bounded expansion to classes of structures over an
arbitrary fixed signature, by requiring that their class of Gaifman
graphs has bounded expansion. In particular, a class of colored graphs
has bounded expansion if and only if the class of underlying uncolored
graphs has bounded expansion.  For an in-depth study of classes with
bounded expansion we refer the reader to the monography
\cite{Sparsity}.

\paragraph*{Nowhere dense classes.}

A class $\Cc$ is {\em nowhere dense} if there is a function
$f\colon \mathbb{N} \rightarrow \mathbb{N}$ such that
$\omega(H)\leq f(r)$ for every $r \in \mathbb{N}$ and every depth-$r$
topological minor $H$ of a graph from $\Cc$ \cite{ND_logic,
  ND_characterization}.

\subsection{Monadic stability, monadic dependence, and low VC-density}
\smallskip
\noindent The model theoretic approach of complexity is based on the
study of properties rather than on the study of objects.  This is
witnessed by the fact that the central subjects of study in model
theory are theories and that the actual structures are only considered
as models of theories.  Nevertheless, most notions defined on theories
have their counterpart on models or on classes of models. One of the
main goals of stability theory (also known as classification theory)
is to classify the models of a given first-order theory according to
some simple system of cardinal invariants.  In this respect,
elementary theories are {\em stable theories} and still reasonably
well behaved theories are {\em NIP theories} (also called {\em
  dependent theories}).  These notions can be translated to classes of
structures as follows:
\begin{definition}
  A class $\Cc$ of structures is {\em stable} if for every first-order
  formula $\varphi(\bar x,\bar y)$ there exists an integer $k$ such
  that for every structure $\mathbf{A}\in\Cc$ and for all tuples
  $\bar a_1,\dots,\bar a_\ell,\bar b_1,\dots,\bar b_\ell$ of elements
  of $\mathbf A$, if
  \begin{equation}
    \mathbf A\models\varphi(\bar a_i,\bar b_j)\quad\iff\quad i<j
  \end{equation}
  for all $i,j\in [\ell]$, then $\ell\leq k$.
\end{definition}

\begin{definition}
  A class $\Cc$ of structures is {\em dependent} (or {\em NIP}) if for
  every first-order formula $\varphi(\bar x,\bar y)$ there exists an
  integer $k$ such that for every structure $\mathbf A\in\Cc$ and for
  all tuples $\bar a_i$ ($i\in[\ell]$) and, $\bar b_I$
  ($I\subseteq [\ell]$) of elements of $\mathbf A$, if
  \begin{equation}
    \mathbf A\models\varphi(\bar a_i,\bar b_I)\quad\iff\quad i\in I
  \end{equation}
  for all $i\in [\ell]$ and all $I\subseteq [\ell]$, then
  $\ell\leq k$.
\end{definition}

A stronger notion of stability and of dependence arises when one
allows to apply arbitrary monadic lifts to the structures in $\Cc$
before using the formula~$\varphi$. These variants are called {\em
  monadic stability} and {\em monadic dependence}. The expressive
power gained by the monadic lift is so strong that tuples of free
variables can be replaced by single free variables in the above
definitions \cite{baldwin1985second}.

\begin{definition}
  A class $\Cc$ of structures is {\em monadically stable} if for every
  first-order formula $\varphi(x,y)$ there exists an integer $k$ such
  that for every monadic lift~$\mathbf A^+$ of a structure
  $\mathbf{A}\in\Cc$ and for all elements
  $a_1,\dots,a_\ell,b_1,\dots,b_\ell$ of $\mathbf A$, if
  \begin{equation}
    \mathbf A^+\models\varphi(a_i,b_j)\quad\iff\quad i<j
  \end{equation}
  for all $i,j\in [\ell]$, then $\ell\leq k$.
\end{definition}

\begin{definition}
  A class $\Cc$ of structures is {\em monadically dependent} (or {\em
    monadically NIP}) if for every first-order formula~$\varphi(x,y)$
  there exists an integer $k$ such that for every monadic lift
  $\mathbf A^+$ of a structure $\mathbf A\in\Cc$ and for all elements
  $a_i$ ($i\in[\ell]$) and $b_I$ ($I\subseteq [\ell]$) of $\mathbf A$,
  if
  \begin{equation}
    \mathbf A^+\models\varphi(a_i,b_I)\quad\iff\quad i\in I
  \end{equation}
  for all $i\in [\ell]$ and all $I\subseteq [\ell]$, then
  $\ell\leq k$.
\end{definition}

For a formula $\phi(\bar x,\bar y)$, the {\em VC-density}
$\mathrm{vc}^\Cc(\phi)$ of a formula $\phi$ in a class $\Cc$
(containing arbitrarily large structures) is defined as
\[
\mathrm{vc}^{\Cc}(\phi)=\lim_{t\rightarrow\infty} \sup_{\mathbf
  A\in\Cc}\sup_{\substack{B\subseteq V(\str A)\\|B|= t}}\frac{\log
  |\{\phi(\bar v, \str A)\cap B^{|\bar x|}\mid \bar v\in V(\str
  A)^{|\bar y|}\}|}{\log |B|}
\]
The {\em VC-density} $\mathrm{vc}^{\Cc}$ of the class $\Cc$ is
\[
\mathrm{vc}^{\Cc}(n)=\sup \{\mathrm{vc}^{\Cc}(\phi)\mid \phi(\bar
x;\bar y) \text{ is a formula with }|\bar y|=n\}.
\]

According to the Sauer-Shelah Lemma
\cite{Sauer1972,shelah1972combinatorial}, a class $\Cc$ is NIP if and
only if $\mathrm{vc}^{\Cc}(\phi)<\infty$ for every formula
$\phi$. However, it is possible for a NIP class (and even for a stable
class) to have $\mathrm{vc}^{\Cc}(1)=\infty$. On the other hand, is
easily checked that (unless structures in $\Cc$ have bounded size) for
every positive integer $n$ we have $\mathrm{vc}^{\Cc}(n)\geq n$. A
class $\Cc$ has {\em low VC-density} if $\mathrm{vc}^{\Cc}(n)=n$ for
all integers~$n$~\cite{guingona2013vc}. We say that $\Cc$ has {\em
  monadically low VC-density} if every monadic lift of $\Cc$ has low
VC-density.

\begin{theorem}\label{thm:Adler}
  Let $\Cc$ be a class of graphs.
  \begin{enumerate}
  \item If $\Cc$ is nowhere dense, then $\Cc$ is monadically stable
    ([Adler, Adler \cite{adler2014interpreting}; Podewski, Ziegler
    \cite{Podewski1978}).
  \item If $\Cc$ is nowhere dense, then $\Cc$ has monadicallly low
    VC-density (Pilipczuk, Siebertz, and Toru{\'n}czyk
    \cite{pilipczuk2018number}]).
  \end{enumerate}
\end{theorem}

\begin{theorem}[[Adler, Adler \cite{adler2014interpreting}; Podewski,
  Ziegler \cite{Podewski1978}] Let $\Cc$ be a monotone class of
  graphs. If $\Cc$ is NIP, then $\Cc$ is nowhere dense.
\end{theorem}

\pagebreak
\begin{corollary}
  Let $\Cc$ be a monotone class of graphs. Then the following are
  equivalent.
  \begin{enumerate}
  \item $\Cc$ is nowhere dense,
  \item $\Cc$ is stable,
  \item $\Cc$ is monadically stable,
  \item $\Cc$ is NIP,
  \item $\Cc$ is monadically NIP,
  \item $\Cc$ has low VC-density,
  \item $\Cc$ has monadically low VC-density.
  \end{enumerate}
\end{corollary}

\subsection{Interpretations and transductions}
\smallskip
\noindent In this paper, by an {\em interpretation} of
$\Sigma'$-structures in $\Sigma$-structures we mean a transformation
$\mathsf I$ defined by means of formulas $\phi_R(\bar x)$ (for
$R\in\Sigma'$ of arity~$|\bar x|$) and a formula $\nu(x)$.  For every
$\Sigma$-structure $\mathbf A$, the $\Sigma'$-structure
$\mathsf I(\mathbf A)$ has domain~$\nu(\mathbf A)$ and the
interpretation of each relation $R\in\Sigma'$ is given by
\mbox{$R^{\mathsf I(\mathbf A)}=\phi_R(\mathbf A)\cap \nu(\mathbf
  A)^{|\bar x|}$}.

A {\em transduction} $\mathsf T$ is the composition
$\mathsf I\circ\Lambda$ of a monadic lift and an interpretation. It is
easily checked that the composition of two transductions is again a
transduction.

Let $\Cc$ and $\Dd$ be classes of $\Sigma_\Cc$-structures and
$\Sigma_\Dd$-structures, respectively. Let $\mathsf{I}$ be an
interpretation of $\Sigma_\Dd$-structures in
$\Sigma_\Cc^+$-structures, where $\Sigma_\Cc^+\setminus\Sigma_\Cc$ is
a finite set of unary relation symbols.  If, for every $\mathbf B$ in
$\Dd$ there exists a lift~$\mathbf A^+$ of some structure
$\mathbf A\in\Cc$ such that $\mathbf B=\mathsf{I}(\mathbf A^+)$ we
write \vspace{-3mm}
\[
\xy\xymatrix{ \Cc \ar@{->>}[r]^(.4){\mathsf{I}}&\Dd, }\endxy
\]
and we write
\[
\xy\xymatrix{ \Cc \ar@{->>}[r]&\Dd }\endxy
\] if there exists $\interp I$ such that
$\xy\xymatrix{ \Cc \ar@{->>}[r]^(.4){\mathsf{I}}&\Dd.  }\endxy$ 
Let
$\mathscr{H}$ denote the class of half graphs and let $\Gg$ denote the
class of all finite graphs.  We have \vspace{-3mm}
\begin{align*}
\Cc\text{ is monadically stable}\quad&\iff\quad
\xymatrix{
	\Cc\ar@{->>}[r]|-(.4){/}&\mathscr{H}.\\
	}\\
\Cc\text{ is monadically NIP}\quad&\iff\quad
\xymatrix{
	\Cc\ar@{->>}[r]|-(.4){/}&\Gg.\\
	}\\
\end{align*}

\vspace{-8mm}

\begin{lemma}[\cite{anderson1990tree}]
  A stable class $\Cc$ is monadically unstable if and only $\Cc$ has a
  transduction to the class of all \mbox{$1$-subdivided} complete
  bipartite graphs.
\end{lemma}
\begin{corollary}\label{crl:mon-stable}
  A class $\Cc$ is monadically stable if and only if it is both stable
  and monadically NIP.
\end{corollary}

\begin{figure}
\begin{center}
  \includegraphics[width=.35\textwidth]{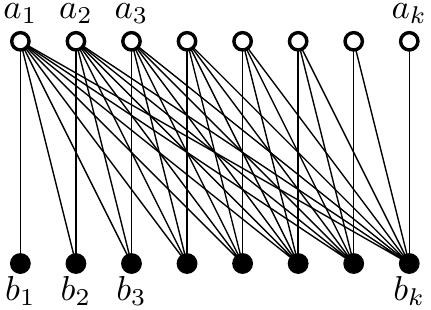}
\end{center}
\caption{The half-graph $H_k$}
\label{fig:HG}
\end{figure}

We use the term of {\em structurally {\rm xxx}} for classes that are
transductions of classes that are xxx. For instance, a class has
structurally bounded treewidth if it is the transduction of a class
with bounded treewidth.

The following characterizations of classes with bounded treewidth,
pathwidth, rankwidth, linear rankwidth, and shrubdepth show the deep
connections between these width measures and logical transductions
(and at this point will serve as a definition of the notions of
rankwidth, linear rankwidth and shrubdepth).

\begin{empty} 
  \begin{enumerate}
  \item A class $\Cc$ of graphs has bounded treewidth (pathwidth,
    respectively) if and only if there exists an
    MSO-transduction~$\mathsf{T}$ such that the incidence graph of
    every $G\in\Cc$ is the result of applying $\mathsf{T}$ to some
    tree (path, respectively) (\cite{courcelle1992monadic} (see also
    \cite{courcelle2012graph}, Theorem 7.47)).
  \item A class $\Cc$ of graphs has bounded rankwidth (linear
    rankwidth, respectively) if and only if there exists an
    MSO-transduction~$\mathsf{T}$ such that every $G\in\Cc$ is the
    result of applying $\mathsf{T}$ to some tree (path, respectively).
    (\cite{courcelle1992monadic} (see also~\cite{courcelle2012graph},
    Theorem 7.47)).
  \item\label{it:colc} A class $\Cc$ of graphs has bounded rankwidth
    (linear rankwidth, respectively) if and only if there exists an
    FO-transduction~$\mathsf{T}$ such that every $G\in\Cc$ is the
    result of applying $\mathsf{T}$ to some tree order (linear order,
    respectively)~(\cite{colcombet2007combinatorial}).
  \item\label{it:sd} A class $\Cc$ of graphs has bounded shrubdepth if
    and only if there exists an FO-transduction~$\mathsf{T}$ and a
    height $h$ such that every $G\in\Cc$ is the result of applying
    $\mathsf{T}$ to some tree of depth at most $h$~(\cite{Ganian2012,
      Ganian2017}).
  \end{enumerate}
\end{empty}

We can rewrite properties~\eqref{it:colc} and~\eqref{it:sd} as follows:

\vspace{-3mm}
\begin{align*}
\Cc\text{ has bounded rankwidth}\quad&\iff\quad
\xymatrix{
	\mathscr{Y}^{\leq}\ar@{->>}[r]&\Cc,\\
	}\\
\Cc\text{ has bounded linear rankwidth}\quad&\iff\quad
\xymatrix{
	\mathscr{L}^{\leq}\ar@{->>}[r]&\Cc,\\
	}\\\Cc\text{ has bounded shrubdepth}\quad&\iff\quad
\exists n\ \xymatrix{
	\mathscr{Y}_n\ar@{->>}[r]&\Cc,\\
	}\\
\end{align*}
where $\mathscr{Y}^{\leq}$ denotes the class of all finite tree
orders, $\mathscr{L}^{\leq}$ denotes the class of all linear orders,
and $\mathscr{Y}_n$ denotes the class of trees with depth at most $n$.

Note that in the characterizations above $\mathscr{Y}^{\leq}$ can be
replaced by the class of trivially perfect graphs (or by the larger
class of cographs) and $\mathscr{L}^{\leq}$ can be replaced by the
class of transitive tournaments or by the class of half-graphs.

\begin{remark}
  Since the class of all graphs does not have bounded rankwidth, we
  deduce that if $\Cc$ has bounded rankwidth we have
  $\xymatrix{ \Cc\ar@{->>}[r]|-(.4){/}&\Gg.  }$ Hence every class with
  bounded rankwidth is monadically NIP.
\end{remark}

In particular, \Cref{crl:mon-stable} implies the following:
\begin{remark}
  A class with bounded rankwidth is monadically stable if and only if
  it is stable.
\end{remark}

\subsection{Weakly sparse classes}

\smallskip\noindent It appears that a basic property that makes a
graph class dense is that graphs in it contain arbitrarily large
bicliques. Indeed, forbidding a biclique as a subgraph (or,
equivalently, forbidding a clique and a biclique as induced subgraphs)
is known to have a strong consequence on classes with low
complexity. We call a class $\Cc$ {\em weakly sparse} if it excludes
some biclique as a subgraph.

\begin{theorem}
Let $\Cc$ be a weakly sparse class of graphs. 
\begin{enumerate}
\item If $\Cc$ has bounded shrubdepth, then $\Cc$ has bounded
  treedepth~\cite{SBE_drops}.
\item If $\Cc$ has bounded linear rankwidth, then $\Cc$ has bounded
  pathwidth~\cite{Gurski2000}.
\item If $\Cc$ has bounded rankwidth, then $\Cc$ has bounded
  treewidth~\cite{Gurski2000}.
\end{enumerate}
\end{theorem}

We call a class {\em sparsifiable} if it is transduction-equivalent to
a weakly sparse class.

\pagebreak

Importance of weakly sparse classes are witnessed by numerous
result. Among them, let us cite
\begin{itemize}
\item The $k$-Dominating Set problem is fixed parameter tractable
  (FPT) and has a polynomial kernel for any weakly sparse class
  \cite{philip2009solving}.
\item Connected $k$-Dominating Set, Independent $k$-Dominating Set and
  Minimum Weight $k$-Dominating Set are FPT, when parameterized by
  $t+k$ (where $t$ is the output size) \cite{telle2012fpt}.
\item Dominating Set Reconfiguration is FPT on weakly sparse classes
  \cite{LOKSHTANOV2018122}.
\item For every graph $H$ and for weakly sparse class $\Cc$ there
  exists $d\in\mathbb N$ such that every graph $G\in\Cc$ with average
  degree at least $d$ contains an induced subdivision of $H$
  \cite{kuhn2004induced}. This result has further been strengthened as
  follows: every weakly sparse class that excludes an induced
  subdivision of some graph $H$ has bounded expansion
  \cite{DVORAK2018143}.
\end{itemize}

The assumption that a class is weakly sparse allows frequently to work
with induced subgraph instead of subgraphs. For instance:

\begin{theorem}[Dvo\v r\'ak \cite{DVORAK2018143}]
  A hereditary weakly sparse class $\Cc$ has bounded expansion if and
  only if there exists a function
  $f\colon\mathbb N\rightarrow\mathbb N$ such that for every graph
  $H$, if the $\leq k$-subdivision of $H$ belongs to $\Cc$ then the
  average degree of~$H$ is at most $f(k)$.
\end{theorem}

We now prove a similar characterization of nowhere dense classes.
\begin{theorem}
  A hereditary weakly sparse class $\Cc$ is nowhere dense if and only
  if there exists a function $f\colon\mathbb N\rightarrow\mathbb N$
  such that the class $\Cc$ contains no $\leq k$-subdivided clique of
  order greater than $f(k)$.
\end{theorem}

This theorem directly follows from the next lemma.

\begin{lemma}
  For all integers $t,p,n$ there exists an integer $N$ such that if a
  graph $G$ contains no $K_{t,t}$ as a subgraph and no induced
  $q$-subdivision of $K_{4t}$ (for any $q\leq p$), then it contains no
  $\leq p$-subdivision of $K_N$ as a subgraph.
\end{lemma}
\begin{proof}
  Assume that $G$ contains no $K_{t,t}$ as a subgraph but contains a
  $\leq p$-subdivision of a large complete graph $K_N$ as a
  subgraph. We can first assume by Ramsey's theorem that $G$ contains
  an exact $q$-subdivision of $K_N$ (for some $q\leq p$). Of course
  $q>1$, for otherwise the ``subdivision'' is induced. Also we can
  assume that each branch of the subdivision is an induced path (for
  otherwise we consider a shorter path).
	
  Let $v_1,\dots,v_N$ be the principal vertices of the $K_N$, and let
  $u_{i,j,k}$ (for $1\leq i<j\leq N$ and $1\leq k\leq q$) be the $k$th
  vertex on the path of length $q+1$ linking $v_i$ to $v_j$ in the
  considered $q$-subdivision of $K_N$.  To every $3$-tuple $(a,b,c)$
  (resp. every $4$-tuple $(a,b,c,d)$) of distinct integers in $[N]$
  (with $a<b<c$, resp. $a<b<c<d$) we associate its type, which is the
  isomorphism type of the (vertex ordered) graph induced by
  $v_a,v_b,v_c$ (resp.\ $v_a,v_b,v_c,v_d$) and the paths of length
  $q+1$ linking these vertices.  By Ramsey's theorem, assuming $N$ is
  sufficiently large, we can extract a subset $X$ of order $4t$ of
  $[N]$, such that all the types of $3$-tuples of elements in $X$ are
  the same and that all the types of $4$-tuples of elements in $X$ are
  the same.  We partition $X$ into $4$ subsets $A,B,C,D$ of order $t$,
  with elements in~$A$ smaller than those in $B$ smaller than those in
  $C$ smaller than those in $D$.
	
  Assume that the type of $3$-tuples is not a cycle. Without loss of
  generality, the type of $(1,2,3)$ contains an edge $v_1,u_{2,3,a}$
  or an edge $u_{1,2,a},u_{1,3,b}$.
  \begin{center}
    \includegraphics[height=.15\textheight]{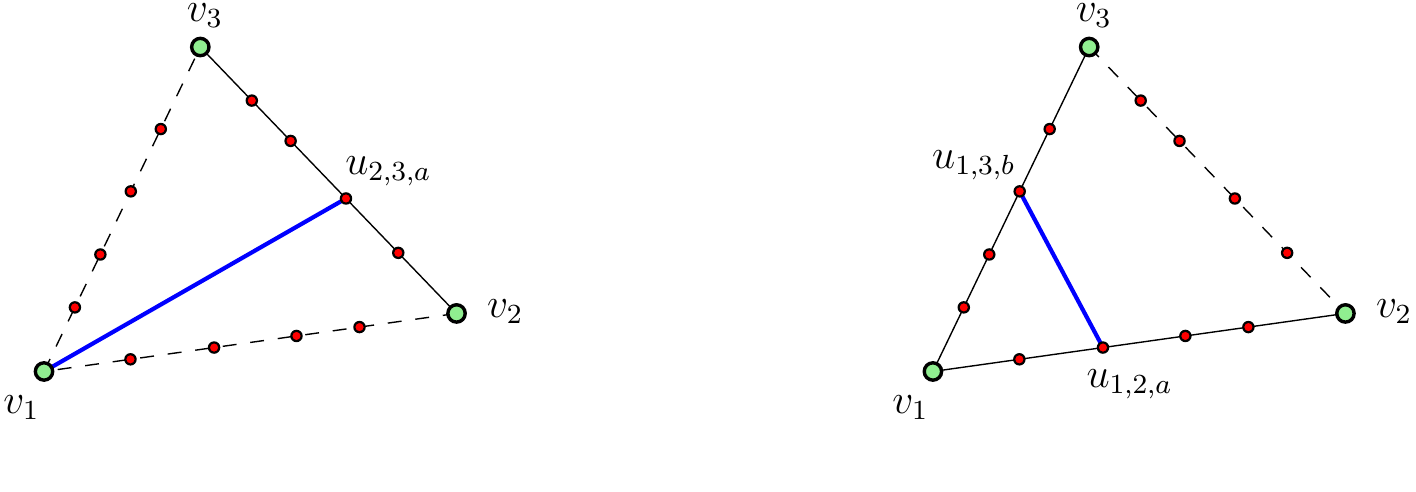}
	\end{center}
	In the first case, choose independently $i\in A$ and $j\in B$
        and fix $k\in C$. Then the vertices $v_i$ and $u_{j,k,a}$
        define a $K_{t,t}$-subgraph.  In the second case, fix $i\in A$
        and choose independently $j\in B$ and $k\in C$. Then the
        vertices~$u_{i,j,a}$ and~$u_{i,k,b}$ define a large complete
        bipartite subgraphs and we conclude as above.
	
        \pagebreak In the case the type of $4$-tuples is not the
        $q$-subdivision of a $K_4$ and that the type of every
        $3$-tuple is a cycle, we can assume without loss of generality
        that the type of $(1,2,3,4)$ contains an edge
        $u_{1,2,a} u_{3,4,b}$.
        \begin{center}
          \includegraphics[height=.15\textheight]{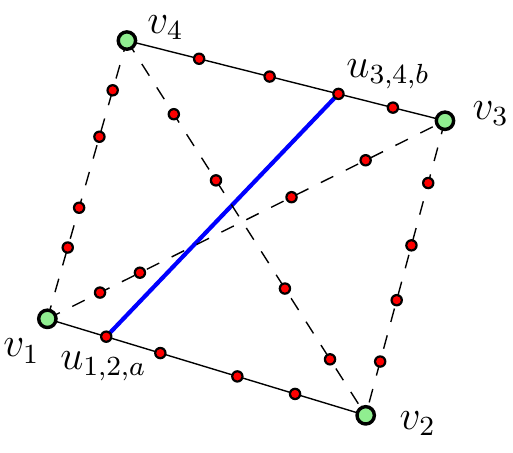}
	\end{center}
	Fix $i\in A$ and $\ell\in D$ and let $j\in B$ and $k\in
        C$.
        Then the vertices $u_{i,j,a}$ and $u_{k,\ell,b}$ define a
        $K_{t,t}$ subgraph.
	
	We deduce that the $q$-subdivision of the clique $K_{4t}$
        defined by $X$ is induced.
\end{proof}

\begin{corollary}
  Let $\Cc$ be a monadically NIP class.  Then $\Cc$ is nowhere dense
  if and only if it is weakly sparse.
\end{corollary}
\begin{proof}
  Assume towards a contradiction that the class $\Cc$ weakly sparse
  and not not nowhere dense. Then there is an integer $p$ such that we
  can find in graphs in $\Cc$ some $\leq p$-subdivisions of
  arbitrarily large cliques. According to the previous lemma we can
  find arbitrarily large induced $q$-subdivisions of cliques for some
  $1<q\leq p$. It is then easy to interpret (in a monadic lift)
  arbitrary graphs, contradicting the hypothesis that~$\Cc$ is
  monadically~NIP.
	\end{proof}

\begin{corollary}
  Every sparsifiable monadically NIP class of graphs is structurally
  nowhere dense.
\end{corollary}

\begin{figure}
	\begin{center}
	\includegraphics[width=.6\textwidth]{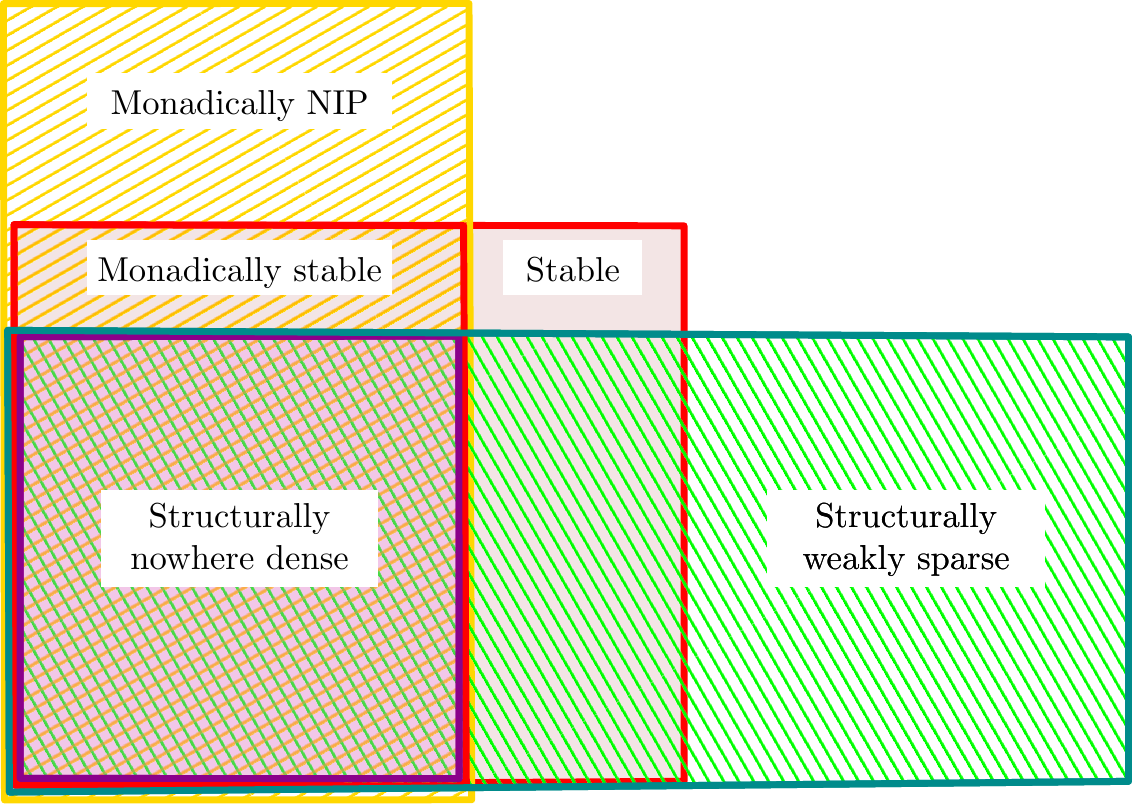}	
        \caption{A class is monadically stable if and only if it is
          both monadically NIP and stable; it is structurally nowhere
          dense if and only if it is both monadically NIP and
          structurally weakly sparse. No class is currently known,
          which is monadically stable but not structurally nowhere
          dense.}
	\label{fig:Univ}
	\end{center}
\end{figure}

\subsection{Decompositions and covers}

\smallskip \noindent For $p \in \mathbb N$, a \emph{$p$-cover} of a
structure $\mathbf A$ is a family $\cal U_{\mathbf A}$ of subsets of
$V(\mathbf A)$ such that every set of at most $p$ elements
of~$\mathbf A$ is contained in some $U \in \cal U_{\mathbf A}$.
If~$\Cc$ is a class of structures, then a {\em $p$-cover} of $\Cc$ is
a family $\cal U=(\cal U_{\mathbf A})_{\mathbf A\in \Cc}$,
where~$\cal U_{\mathbf A}$ is a $p$-cover of $\mathbf A$.  A $1$-cover
is simply called a {\em cover}.  A $p$-cover $\cal U$ is \emph{finite}
if $\sup\{|\cal U_{\mathbf A}|\mid \mathbf A\in \Cc\}$ is finite.
Let~$\Cc[\cal U]$ denote the class structures
$\{\mathbf A[U]\mid \mathbf A\in\Cc, U\in \cal U_{\mathbf A}\}$. For a
class $\Ww$ we say that a cover $\cal U$ is a $\Ww$-cover if
$\Cc[\cal U]\subseteq \Ww$. If $\Ww$ is a class of bounded treedepth,
bounded shrubdepth, etc., we call a $\Ww$-cover a bounded treedepth
cover, bounded shrubdepth cover, etc.  The class $\Cc$ admits
\emph{low treedepth covers}, \emph{low shrubdepth covers}, etc.\ if
and only if for every $p\in \N$ there is a finite $p$-cover
$\mathcal{U}_p$ of $\Cc$ with bounded treedepth, shrubdepth, etc.

\begin{theorem}[\cite{POMNI,SBE_drops}]
  A class of graphs has bounded expansion if and only if it has low
  treedepth covers.
\end{theorem}

The following notion of \emph{shrubdepth} has been proposed
in~\cite{Ganian2012} as a dense analogue of treedepth.  Originally,
shrubdepth was defined using the notion of \emph{tree-models}. We
present an equivalent definition based on the notion of
\emph{connection models}, introduced in~\cite{Ganian2012} under the
name of {\em{$m$-partite cographs}} with bounded depth.

A {\em connection model} with labels from $\Gamma$ is a rooted labeled
tree $T$ where each leaf $u$ is labeled by a label
$\gamma(u)\in \Gamma$, and each non-leaf node $x$ is labeled by a
binary relation $C(x)\subset \Gamma\times \Gamma$.  If $C(x)$ is
symmetric for all non-leaf nodes $x$, then such a model defines a
graph $G$ on the leaves of $T$, in which two distinct leaves~$u$ and
$v$ are connected by an edge if and only if
$(\gamma(v),\gamma(v))\in C(x)$, where~$x$ is the least common
ancestor of $u$ and $v$.  We say that $T$ is a \emph{connection model}
of the resulting graph $G$.  A class of graphs $\Cc$ has {\em bounded
  shrubdepth} if there is a number $h\in \mathbb N$ and a finite set
of labels~$\Gamma$ such that every graph $G\in \Cc$ has a connection
model of depth at most $h$ using labels from $\Gamma$.

A \emph{cograph} is a graph that has a connection model (called a
\emph{cotree}) with a labels set $\Gamma$ containing only a single
label. Cographs are perfect graphs, that is, graphs in which the
chromatic number of every induced subgraph equals the clique number of
that subgraph.

\begin{theorem}[\cite{SBE_drops}]
	  A class of 
   graphs has structurally bounded expansion if and only
   if it has low shrubdepth covers.
\end{theorem}

The {\em c-chromatic number} of a graph $G$ is the minimum size of a partition $V_1,\dots,V_k$ of the vertex set of $G$ such that~$G[V_i]$ is a cograph for each $i\in\{1,\dots,k\}$. We denote by $\chi_c(G)$ the c-chromatic number of $G$.

\begin{lemma}
  Every class with bounded shrubdepth has bounded c-chromatic number.
\end{lemma}
\begin{proof}
  Let $h\in\mathbb N$ and let $\Gamma$ be a finite set such that every
  graph $G\in\Cc$ has a connection model of depth at most $h$ using
  labels from $\Gamma$, and let $\alpha\in\Gamma$. It is easily
  checked that the subgraph of $G$ induced by the vertices with label
  $\alpha$ has a connection model using only the label $\alpha$. It
  follows that this induced subgraph is a cograph, hence the
  c-chromatic number of $G$ is at most $|\Gamma|$.
\end{proof}

\begin{corollary}
  Every class $\Cc$ that admits $1$-covers of bounded shrubdepth has
  bounded c-chromatic number, and hence is linearly $\chi$-bounded.
\end{corollary}

\begin{lemma}[\cite{SBE_drops}]
  Every class that admits $2$-covers of bounded shrubdepth is
  sparsifiable.
\end{lemma}

\section{Rankwidth and linear rankwidth}
\smallskip\noindent We now turn to the study of classes of bounded
rankwidth and linear rankwidth. After recalling several equivalent
definitions of these width measures, we prove for every proper
hereditary family $\Ff$ of graphs (like cographs) that there is
a class with bounded rankwidth that does not have the property that
graphs in it can be colored by a bounded number of colors, each
inducing a subgraph in $\Ff$.

\subsection{Definitions}
\smallskip\noindent Classes with bounded rankwidth and classes with
bounded linear rankwidth enjoy several characterizations.  In
par\-ticular, for a class $\Cc$ the following are equivalent:
\begin{enumerate}
	\item $\Cc$ has bounded rankwidth,
	\item $\Cc$ has bounded cliquewidth,
	\item $\Cc$ has bounded NLC-width,
	\item $\xymatrix{\mathscr{Y}^{\leq}\ar@{->>}[r]&\Cc}$,
\end{enumerate}
as well as the following:
\begin{enumerate}
	\item $\Cc$ has bounded linear rankwidth,
	\item $\Cc$ has bounded linear cliquewidth,
	\item $\Cc$ has bounded linear NLC-width,
	\item $\Cc$ has bounded neighborhood-width,
	\item $\xymatrix{\mathscr{L}^{\leq}\ar@{->>}[r]&\Cc}$.
\end{enumerate}

\paragraph*{Cliquewidth and linear cliquewidth.}
Graphs of bounded treewidth have bounded average degree and therefore
the application of treewidth is (mostly) limited to sparse graph
classes. \emph{Cliquewidth} was introduced in
\cite{courcelle1993handle} with the aim to extend hierarchical
decompositions also to dense graphs. However, there is no known
polynomial-time algorithm to determine whether the cliquewidth of an
input graph is at most $k$ for fixed $k\geq 4$.  A notable application
of cliquewidth is the extension of Courcelle's Theorem for testing MSO
properties in cubic time (or linear time if a clique decomposition is
given) on graph classes of bounded
cliquewidth~\cite{CourcelleMakRot00}.  The notion of linear
cliquewidth has been introduced in \cite{gurski2005relationship}. We
denote by $\mathrm{cw}(G)$ the cliquewidth of a graph $G$ and by
$\mathrm{lcw}(G)$ the linear cliquewidth of $G$.

\paragraph*{NLC-width and linear NCL-width.}
The notions of NLC-width and linear NLC-width were introduced
in~\cite{wanke1994k} and~\cite{gurski2005relationship}.  Let $k$ be
some positive integer. We are going to work with the following
definition of linear NLC-width.
\begin{definition}\label{def:Sigmak}
  For $k\in\N$, let $V$ be a finite set, and let $\Omega_k(V)$ be the
  alphabet whose letters are quadruples~$(v, c,e,r)$, where
  \begin{itemize}
  \item $v\in V$,
  \item $c\in [k]$,
  \item $e\subseteq [k]$, and
  \item $r\colon [k]\rightarrow[k]$.
  \end{itemize}
  
  For a letter $a=(v,c,e,r)\in \Omega_k(V)$ we write $v_a, c_a, e_a$
  and $r_a$ for $v,c,e$ and~$r$, respectively.
\end{definition}

We say that a word $\alpha\in\Omega_k(V)^+$ is {\em admissible} if no
two letters $a$ and $b$ of $\alpha$ have the same $v$-value. We denote
by~$\mathfrak L_k(V)$ the set of all admissible words in~$\Omega_k^+$.

\begin{definition} \label{def:NLC-width} 
  A \emph{linear NLC-expression
    of width $k$ over $V$} is a word in $\mathfrak L_k(V)$.  With
  linear NLC-expressions $\alpha$ of width $k$ over $V$ we recursively
  associate a colored graph $\Xi(\alpha)$ whose vertices are the
  $v$-values of the letters of $\alpha$, colored by colors from $[k]$
  as follows.
  \begin{itemize}
  \item If $|\alpha|=1$, then $\Xi(\alpha)$ is the single vertex
    graph, with vertex $v_\alpha$ colored $c_\alpha$.
  \item If $\alpha=\alpha'a$, where $|a|=1$, then $\Xi(\alpha)$ is the
    graph obtained from $\Xi(\alpha')$ by adding the vertex $v_a$ with
    color $c_a$, connecting $v_a$ to all vertices $w\in \Xi(\alpha')$
    that have a color in $e_a$, and finally, changing the color of
    each vertex with color $i$ to color $r_a(i)$.
  \end{itemize}
  The linear NLC-width of a graph $G$ is the minimum integer $k$ such
  that $G$ is identical to the graph $\Xi(\alpha)$ for some
  $\alpha\in\mathfrak L_k(V(G))$.
\end{definition}

It is clear that the vertex set of $\Xi(\alpha)$ can be identified
with the letters of $\alpha$.  and that for every subword $\beta$ of
$\alpha$ the graph $\Xi(\beta)$ is the subgraph of $\Xi(\alpha)$
induced by the $v$-values of the letters of $\beta$.
%
%
%
We have \cite{gurski2005relationship}:
\begin{equation}
  \text{linear NLC-width}(G)\leq \mathrm{lcw}(G)\leq \text{linear NLC-width}(G)+1.
\end{equation}

\paragraph*{Neighborhood-width.}
The {\em neighborhood-width} of a graph is the smallest integer $k$,
such that there is a linear order $v_1,\ldots,v_n$ on the vertex set
of $G$ such that for every vertex $v_j$ the vertices $v_i$ with
$i\leq j$ can be divided into at most $k$ subsets, each members having
the same neighborhood with respect to the vertices $v_k$ with $k>j$.
The neighbourhood-width of a graph differs from its linear
clique-width or linear NLC-width at most by one \cite{GURSKI20061637}.

\paragraph*{Rankwidth and linear rankwidth.}
The notion of~\emph{rankwidth} was introduced
in~\cite{oum2006approximating} as an efficient approximation to
cliquewidth.  For a graph $G$ and a subset $X\subseteq V(G)$ we define
the \emph{cut-rank} of $X$ in $G$, denoted~$\rho_G(X)$, as the rank of
the $|X|\times |V(G)\setminus X|$ $0$-$1$ matrix $A_X$ over the binary
field~$\mathbb{F}_2$, where the entry of $A_X$ on the $i$-th row and
$j$-th column is~$1$ if and only if the $i$-th vertex in $X$ is
adjacent to the $j$-th vertex in $V(G)\setminus X$. If $X=\emptyset$
or $X=V(G)$, then we define $\rho_G(X)$ to be zero.

A \emph{subcubic} tree is a tree where every node has degree $1$ or
$3$. A \emph{rank decomposition} of a graph~$G$ is a pair $(T,L)$,
where $T$ is a subcubic tree with at least two nodes and $L$ is a
bijection from $V(G)$ to the set of leaves of $T$.  For an edge
$e\in E(T)$, the connected components of $T-e$ induce a partition
$(X,Y)$ of the set of leaves of $T$. The \emph{width} of an edge~$e$
of $(T,L)$ is $\rho_G(L^{-1}(X))$. The width of $(T,L)$ is the maximum
width over all edges of $T$. The \emph{rankwidth} $\rw(G)$ of $G$ is
the minimum width over all rank decompositions of $G$.

Cliquewidth and rankwidth are functionally
related~\cite{oum2006approximating}: For every graph $G$ we have
\begin{equation}
  \rw(G)\leq\mathrm{cw}(G)\leq 2^{\rw(G)+1}-1.
\end{equation}
Hence, a class $\Cc$ of graphs has bounded cliquewidth if and only if
$\Cc$ has bounded rankwidth.


\smallskip The \emph{linear rankwidth} of a graph is a linearized
variant of rankwidth, similarly as pathwidth is a linearized variant
of treewidth.  Let $G$ be an $n$-vertex graph and let
$v_1,\ldots, v_n$ be an order of $V(G)$. The \emph{width} of this
order is $\max_{1\leq i\leq n-1}\rho_G(\{v_1,\ldots, v_i\})$.  The
\emph{linear rankwidth} of $G$, denoted $\lrw(G)$, is the minimum
width over all linear orders of $G$. If $G$ has less than $2$ vertices
we define the linear rankwidth of $G$ to be zero.  An alternative way
to define the linear rankwidth is to define a linear rank
decom\-posi\-tion $(T,L)$ to be a rank decomposition such that $T$ is
a caterpillar and then define linear rankwidth as the minimum width
over all linear rank decompositions. Recall that a caterpillar is a
tree in which all the vertices are within distance $1$ of a central
path.

It was proved in \cite{GURSKI20061637} that the linear cliquewidth and
the linear rankwidth of a graph are bound to each other: Precisely,
for every graph $G$ we have
\begin{equation}
  \lrw(G)\leq \text{linear NLC-width}(G)\leq\mathrm{lcw}(G)\leq 2^{\lrw(G)}.
\end{equation}

A linear ordering witnessing $\lrw(G)\leq k$ (or deciding
$\lrw(G)>k$) for fixed $k$ can be computed in time
$O(n^3)$~\cite{jeong2017art}.

\subsection{Substitution and lexicographic product}
\smallskip\noindent We denote by $G \bullet H$ the lexicographic
product of $G$ and $H$. Note that this operation, though
non-commutative, is associative. By $G \oplus H$ we denote the
operation of forming the disjoint union of $G$ and $H$ and connecting
all vertices of the copy of $G$ to all vertices of the copy of $H$.

\begin{lemma}
For all graphs $G,H$ we have
\[
\mathrm{rw}((G\bullet H)\oplus K_1)=\max(\mathrm{rw}(G\oplus
K_1),\mathrm{rw}(H \oplus K_1)).
\]	
\end{lemma}
\begin{proof}
  Let $(Y_G, L_G)$ and $(Y_H,L_H)$ be rank decompositions of
  $G\oplus K_1$ and $H\oplus K_1$, respectively, of minimum
  width. Assume the leaves of $Y_G$ are $V(G)\cup\{\alpha\}$ and the
  leaves of $Y_H$ are $V(H)\cup\{\beta\}$.  Consider $|G|$ copies of
  $Y_H$ and glue these copies on $Y_G$ by identifying each leaf of
  $Y_G$ that is a vertex of $G$ with the vertex $\beta$ of the
  associated copy.  The obtained tree $Y$ together with the naturally
  inherited mapping $L$ from the vertices of $(G\bullet H)\oplus K_1$
  to the leaves of $Y$ is a branch-decomposition of
  $(G\bullet H)\oplus K_1$ (see \Cref{fig:rwlex}).
	
  Now consider any edge of this branch-decomposition of
  $(G\bullet H)\oplus K_1$. There are two cases:
  \begin{itemize}
  \item Assume the edge is within the branch-decomposition $Y_G$ of
    $G\oplus K_1$. Let $A,B$ be the induced partition of the vertices
    of $(G\bullet H)\oplus K_1$. This partition corresponds to a
    partition $A',B'$ of $G\oplus K_1$. Let $p:A\rightarrow A'$ be the
    natural projection.  We may assume that the vertex $\alpha$
    belongs to $B$ in $(G\bullet H)\oplus K_1$ (hence to $B'$ in
    $G\oplus K_1$). For every vertex $v\in B$ we have
    $N_{(G\bullet H)\oplus K_1}(v)\cap A=(N_{G\oplus K_1}(p(v))\cap
    A')\times V(H)$.
    Hence the cut-rank of $(A,B)$ in $(G\bullet H)\oplus K_1$ equals
    the cut-rank of $(A',B')$ in $G\oplus K_1$.
  \item Otherwise, the edge is within the branch-decomposition of a
    copy of $H\oplus K_1$. Let $A,B$ be the induced partition of the
    vertices of $(G\bullet H)\oplus K_1$, where
    $B\subseteq \{v_0\}\times B'$ for some $v_0\in V(G)$ and some
    $B'\subseteq V(H)$. Then all vertices
    $v\in (\{v_0\}\times V(H))\setminus B$ have the neighborhood
    $(\{v_0\}\times N_H(v))\mathbin{\cap} B$ on $B$, while the
    vertices $v\in A\setminus(\{v_0\}\times V(H))$ have the same
    neighborhood in $B$, which is
    $\{v_0\}\times N_{H\oplus K_1}(\beta)$. It follows that the
    cut-rank of $(A,B)$ in $(G\bullet H)\oplus K_1$ equals the
    cut-rank of $V(H\oplus K_1)\setminus B',B')$ in $V(H\oplus K_1)$.
  \end{itemize}
  It follows that
  $\mathrm{rw}((G\bullet H)\oplus K_1)\leq \max(\mathrm{rw}(G\oplus
  K_1),\mathrm{rw}(H \oplus K_1))$.
  The reverse inequality follows from the fact that $G\oplus K_1$ and
  $H\oplus K_1$ are both induced subgraphs of
  $(G\bullet H)\oplus K_1$.
\end{proof}

\begin{figure}
  \begin{center}
    \includegraphics[width=.8\textwidth]{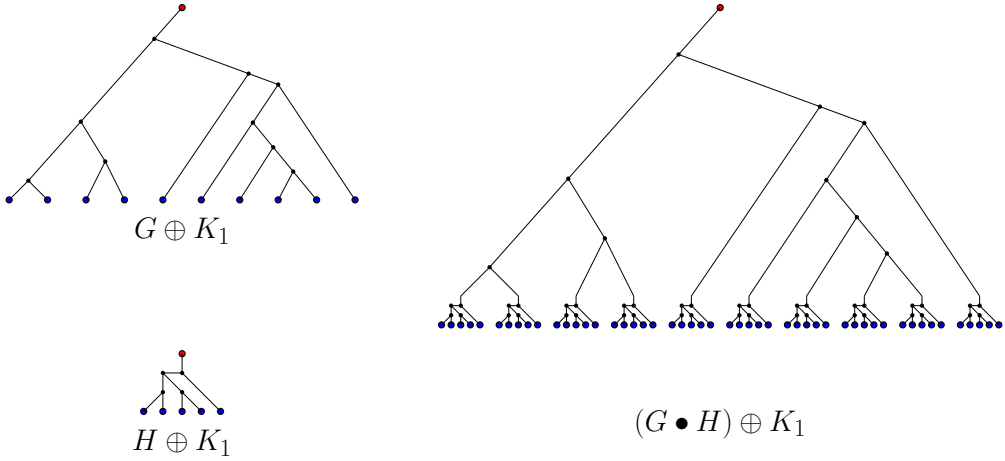}
  \end{center}
  \caption{Branch decomposition of $(G\bullet H)\oplus K_1$ from the
    branch decompositions of $G\oplus K_1$ and $H\oplus K_1$.}
  \label{fig:rwlex}
\end{figure}

Actually the proof of the previous lemma shows that if $G'$ is
obtained from~$G$ by substituting $H$ at some vertex of~$G$, then
\mbox{$\rw(G'\oplus K_1)=\max(\rw(G\oplus K_1)$}, $\rw(H\oplus K_1)$.
(The graph $G\bullet H$ is the substitution of $H$ at every vertex of
$G$).

\begin{corollary}
  Closing a class by substitution increases the rankwidth by at most
  one.
\end{corollary}

For a class $\Cc$, let $\Cc\oplus K_1$ denote the class
$\{G\oplus K_1\mid G\in\Cc\}$, and let $\Cc^{\bullet}$ denote the
closure of $\Cc$ under lexicographic product.  As a direct consequence
of the previous lemma we have
\begin{corollary}
\label{cor:lex}
For every class of graphs $\Cc$ with bounded rankwidth we have
\begin{equation}
  \mathrm{rw}(\Cc)\leq \mathrm{rw}(\Cc^{\bullet})=\mathrm{rw}(\Cc\oplus K_1)\leq \mathrm{rw}(\Cc)+1.
\end{equation}
\end{corollary}
(Indeed, $G\oplus K_1\subseteq_i G\bullet H$ if $H$ contains at least
one edge.)  \medskip

By substituting each vertex of $V(G)$ in the linear order witnessing
\mbox{$\lrw(G\oplus K_1)$} by the linear order of $V(H)$ witnessing
$\lrw(H\oplus K_1$) we similarly obtain the following results.

\begin{lemma}
\label{lem:lrwlex}
For all graphs $G,H$ we have
\[
\lrw(G\bullet H)\leq\lrw(G)+\lrw(H).
\]	
\end{lemma}
\begin{proof}
  Let $<_1$ be a linear order of $V(G)$ witnessing $\lrw(G)$ and let
  $<_2$ be a linear order of $V(H)$ witnessing $\lrw(H)$.  Let $<$ be
  the lexicographic order on $V=V(G)\times V(H)$ defined by $<_1,<_2$,
  i.e., $(u,v)<(u',v')$ if $u<u'$ or ($u=u'$ and $v<v'$).  Let
  $t=(u_t,v_t)\in V$ and let $(u,v)\leq t$. We have
  \[
  \begin{split}
    N_{G\bullet H}((u,v))\cap V^{>t}=\big((N_G(u)\cap
    V(G)^{>u_t})\times V(H)\big) \cup \big(\{u_t\}\times(N_{H}(v)\cap
    V(H)^{>v_t})\big).
  \end{split}
  \]
  It follows that the vector space spanned by the sets
  $N_{G\bullet H}((u,v))\cap V^{>t}$ is in the sum of the vector space
  spanned by the sets $(N_G(u)\cap V(G)^{>u_t})\times V(H)$ (which has
  dimension at most $\lrw(G)$) and of the vector space spanned by the
  sets $\{u_t\}\times(N_{H}(v)\cap V(H)^{>v_t})$ (which has dimension
  at most $\lrw(H)$). Hence the claim follows.
\end{proof}

\subsection{Ramsey properties of rankwidth}

\smallskip\noindent In this section we prove that the class of all
graphs with rankwidth at most $r+1$ is ``Ramsey'' for the class of all
graphs with rankwidth at most $r$, in the following sense.

\begin{theorem}
  For all integers $r,m$ and every graph $G$ with rankwidth at most
  $r$ there exists a graph $G'=G^{\bullet m}$ with rankwidth $r+1$ and
  with the property that every $m$-coloring of $G'$ contains an
  induced monochromatic copy of $G$.
\end{theorem}


\begin{proof}
  We define inductively graph $G^{\bullet i}$ for $i\geq 1$:
  $G^{\bullet 1}=G$ and, for $i\geq 1$ we let
  $G^{\bullet (i+1)}=G^{\bullet i}\bullet G=G\bullet G^{\bullet i}$.
  According to \Cref{cor:lex} we have
  $\mathrm{rw}(\{G^{\bullet i}\mid i\in\mathbb N\})\leq r+1$.
	
  We prove by induction on $m$ that in every $m$-partition of
  $G'=G^{\bullet m}$ one class induces a subgraph with a copy of
  $G$. If $m=1$ the result is straightforward.  Let $m>1$.  Consider a
  partition $V_1,\dots,V_m$ of the vertex set of $G^{\bullet m}$. If
  all the copies of $G^{\bullet (m-1)}$ forming $G^{\bullet m}$
  contain a vertex in $V_m$, then $G^{\bullet m}[V_m]$ contains an
  induced copy of $G$. Otherwise, there is a copy of
  $G^{\bullet (m-1)}$ in $G^{\bullet m}$ whose vertex set is covered
  by $V_1,\dots,V_{m-1}$. By induction hypothesis
  $G^{\bullet (m-1)}[V_i]$ contains an induced copy of $G$.
\end{proof}

\begin{corollary}
\label{cor:herw}
Let $\Ff$ be a proper hereditary class of graphs. Then there exists a
class $\Cc$ with bounded rankwidth such that for every integer $m$
there is $G\in\Cc$ with the property that for every partition of
$V(G)$ into $m$ classes, one class induces a graph not in $\Ff$.
\end{corollary}
%
\begin{corollary}
\label{cor:cog}
The class of graphs with rankwidth at most $2$ does not have the
property that its graphs can be vertex partitioned into a bounded
number of cographs, or circle graphs, etc.
\end{corollary}

\subsection{Lower bounds for $\chi$-boundedness}

\smallskip\noindent Bonamy and Pilipczuk \cite{rw_polychi} announced
independently that classes with bounded rankwidth are polynomially
$\chi$-bounded. We give here a lower bound on the degrees of the
involved polynomials. We write $\chi_f(G)$ for the \emph{fractional
  chromatic number} of a graph $G$, which is defined as
$\chi_f(G)=\inf\bigl\{\frac{\chi(G\bullet K_n)}{n} \mid n\in \N\bigr\}$.

\begin{theorem}
\label{thm:degrw}
For $r\in\mathbb N$, let $P_r$ be a polynomial such that for every
graph $G$ with rankwidth at most $r$ we have
$\chi(G)\leq P_{r}(\omega(G))$.  Then $\deg P_r\in\Omega(\log r)$.
\end{theorem}
\begin{proof}
  As shown in~\cite{geller1975chromatic} for all graphs $G$ and $H$ we
  have $\chi(G\bullet H)= \chi(G\bullet K_{\chi(H)})$. Furthermore we
  have $\chi(G\bullet K_{\chi(H)})\geq \chi(H)\chi_f(G)$.  We deduce
  that $\chi(G\bullet H)\geq \chi_f(G)\chi(H)$.  Hence for every
  integer $n$ we have $\chi(G^{\bullet n})\geq \chi_f(G)^n$. As
  $\omega(G^{\bullet n})=\omega(G)^n$ we have
  $\chi(G^{\bullet n})\geq \omega(G^{\bullet
    n})^{\frac{\log\chi_f(G)}{\log\omega(G)}}$ and hence
\[
\deg P_r
\geq \sup_{\mathrm{rw}(G\oplus K_1)\leq r}\frac{\log \chi_f(G)}{\log
  \omega(G)}.
\]

For sufficiently large integers $n$ there exists a triangle-free graph
$G_n$ with $\chi_f(G_n)\geq\frac{1}{9}\sqrt{\frac{n}{\log n}}$ (see
\cite{R3t}).  As $n>\mathrm{rw}(G_n\oplus K_1)$ we deduce that for
sufficiently large integers $r$ we have
\[
\deg P_r\geq \biggl(\frac{1}{2\log 2}-o(1)\biggr)\log r.
\]
\end{proof}

\paragraph*{Linear rankwidth.}
We give a short proof in \Cref{sec:NLC}
(\Cref{cor:chiLNLC}) that classes with bounded linear
rankwidth are linearly $\chi$-bounded using the equivalence between
classes with bounded linear rankwidth and classes with bounded linear
NLC-width. We improve the obtained upper bound of the $\chi/\omega$
ratio in \Cref{sec:lrw} using a more technical analysis of
linear rank-width (\Cref{thm:cog}), leading to an order of
magnitude of $2^{O(r^2)}$.  We now prove that the ratio $\chi/\omega$
can be as large as $\alpha^r$ for some constant $\alpha>1$ and for
graphs with arbitrarily large linear rankwidth $r$ and clique number
$\omega$.
%
%
%

From \Cref{lem:lrwlex} we deduce $\lrw(C_5^{\bullet n})\leq 2n$.
As $\omega(C_5^{\bullet n})=2^n$ and as
$\chi(C_5^{\bullet n})\geq \chi(C_5)\chi_f(C_5)^{n-1}=3(5/2)^{n-1}$ we
deduce

\[
\frac{\chi(C_5^{\bullet n})}{\omega(C_5^{\bullet n})}\geq
(6/5)(5/4)^{n}\geq (6/5)(5/4)^{\lrw(C_5^{\bullet n})/2}.
\]

As $6/5>\sqrt{5}/2$, for every integer $r$ we have:
\begin{equation}
\label{eq:lrwlb}
\lim_{t\rightarrow\infty}\ \sup_{\substack{\lrw(G)\leq r\\ \omega(G)\geq t}}\ \frac{\chi(G)}{\omega(G)}\geq \biggl(\frac{\sqrt{5}}{2}\biggr)^r.
\end{equation}



\section{Linear NLC-width}
\label{sec:NLC}
\smallskip\noindent In this section we prove that classes with bounded
linear NLC-width (and hence classes of bounded linear rankwdith) are
linearly $\chi$-bounded, and if they are stable, then they are
transduction equivalent to classes of bounded pathwidth. We prove the
result using Simon's factorization forest theorem.

%

\subsection{Simon's factorization forest theorem}
\smallskip\noindent A \emph{semigroup} is an algebra with one
associative binary operation, usually denoted as multiplication.  An
\emph{idempotent} in a semigroup is an element $e$ with
\mbox{$ee=e$}. Given an alphabet $\Omega$ we denote by $\Omega^+$ the
semigroup of all non-empty finite words over $\Omega$, with
concatenation as product.

Fix an alphabet $\Omega$ and a semigroup morphism
$h\colon\Omega^+\rightarrow T$, where $T$ is a finite semigroup.  A
{\em factorization tree} is an ordered rooted tree in which each node
is either a leaf labeled by a letter, or an internal node. The value
of a node is the word obtained by reading the descendant leaves below
from left to right. The value of a factorization tree is the value of
the root of the tree. A {\em factorization tree} of a word
$\alpha\in\Omega^+$ is a factorization tree of value $w$. The {\em
  depth} of the tree is defined as usual, with the convention that the
depth of a single leaf is $1$. A factorization tree is {\em Ramseyan}
(for~$h$) if every node 1) is a leaf, or 2) has two children, or, 3)
the values of its children are all mapped by~$h$ to the same
idempotent of $T$.

\begin{theorem}[Simon's Factorization Forest Theorem
  \cite{kufleitner2008height,simon1990factorization}]
\label{thm:simon}
For every alphabet $\Omega$, every finite semigroup $T$, every
semigroup morphism $h\colon\Omega^+\rightarrow T$, and every word
$\alpha\in\Omega^+$, the word $\alpha$ has a Ramseyan factorization
tree of depth at most $3|T|$.
\end{theorem}
%

The existence of an upper bound expressed only in terms of $|T|$ was
first proved by Simon~\cite{simon1990factorization}.  The improved
upper bound of $3|T|$ is due to
Kufleitner~\cite{kufleitner2008height}.

\subsection{Application to classes with bounded linear NLC-width}

\smallskip\noindent In the following we consider the semigroup
$\Gamma_k$ on functions $r\colon [k] \rightarrow [k]$. Obviously,
\mbox{$h\colon \Omega_k(V)^+\rightarrow \Gamma_k$} induced by
$h(a)=r_a$ for $a\in \Omega_k(V)$ is a semigroup homomorphism (recall
\Cref{def:Sigmak}).  An idempotent of $\Gamma_k$ is a
function $r$ that satisfies that if \mbox{$r(i)=j$}, then $r(j)=j$.
We call $\alpha\in \Omega_k(V)^+$ an idempotent if $h(\alpha)$ is an
idempotent in $\Gamma_k$.

For $\alpha\in \mathfrak L_k(V)$ (recall
\Cref{def:NLC-width}) and for a letter $a$ of $\alpha$ and
$v=v_a$ define $\mathrm{col}_\alpha(v)$ as the color of the vertex $v$
in $\Xi(\alpha)$. Note that if $\alpha\beta\in \mathfrak L_k(V)$ then
$\mathrm{col}_{\alpha\beta}(v)=h(\beta)(\mathrm{col}_\alpha(v))$.

Fix $\alpha\in\mathfrak{L}_k(V)$. According to
\Cref{thm:simon}, there exists a rooted tree $Y$ that is a
Ramseyan factorization tree of $\alpha$ for $h$ with depth at most
$3|T|$.
%
We identify the vertices of $\Xi(\alpha)$ with the leaves of $Y$.  Let
$z$ be a letter of $\Xi(\alpha)$ and let $\beta$ be an ancestor of
$z$. Let $\beta=b_1\dots b_n$ (where the $b_i$ are letters) and let
$p\leq n$ be such that $b_p=z$.  We define
 \begin{align*}
   \mathrm{recol}_\beta(z)&=r_{b_{p-1}}\circ\dots\circ r_{b_1},\\
   \mathrm{eset}_\beta(z)&=\mathrm{recol}_\beta(z)^{-1}(e_z).
 \end{align*}
 \begin{lemma}
   \label{lem:scode}
   Let $z_1,z_2$ be two letters of $\alpha$ appearing in this order in
   $\alpha$, let $\beta$ be their least common ancestor, and
   let~$\delta_1$ (resp.\ $\delta_2$) be the children of $\beta$
   containing the letter $z_1$ (resp.\ $z_2$).  Then $v_{z_1}$ and
   $v_{z_2}$ are adjacent in $\Xi(\alpha)$ if
   \begin{itemize}
   \item $\delta_1$ is not immediately to the left of $\delta_2$ in
     $\alpha$ and
     $\mathrm{col}_\beta(z_1)\in\mathrm{eset}_\beta(z_2)$, or
   \item $\delta_1$ is immediately to the left of $\delta_2$ in
     $\alpha$ and
     $\mathrm{col}_{\delta_1}(z_1)\in\mathrm{eset}_\beta(z_2)$.
   \end{itemize}
 \end{lemma}
 \begin{proof}
   When $\delta_1$ and $\delta_2$ are consecutive, let
   $\delta_2=b_1\dots b_p$ with $b_p=z_2$. Then $v_{z_1}$ and
   $v_{z_2}$ are adjacent if
   \begin{align*}
     &\mathrm{col}_{\delta_1 b_{1}\dots b_{p-1}}(z_1)\in e_{z_2}\\
     \iff\quad&\mathrm{recol}_\beta(\mathrm{col}_{\delta_1}(z_1))\in e_{z_2}\\
     \iff\quad&\mathrm{col}_{\delta_1}(z_1)\in \mathrm{eset}_\beta(z_2).
   \end{align*}
   (Note that in this case we do not make any assumption on
   $h(\delta_1)$ and $h(\delta_2)$.)

   Now assume that $\delta_1$ and $\delta_2$ are non-consecutive. Let
   $\beta_1,\dots,\beta_n$ be the children of $\beta$, and let
   $j\geq i+2$ be such that $\delta_1=\beta_i$ and $\delta_2=\beta_j$.
   As $\beta$ has more than two children, the corresponding
   factorization is a factorization into idempotents.  Let
   $r=h(\beta_1)=\dots=h(\beta_n)$.  Let $\delta_2=b_1\dots b_p$ with
   $b_p=z_2$.

   Then $v_{z_1}$ and $v_{z_2}$ are adjacent if
   \begin{align*}
     &\mathrm{col}_{\beta_i\dots\beta_{j-1}b_{1}\dots b_{p-1}}(z_1)\in e_{z_2}\\
     \iff\quad&\mathrm{recol}_\beta(\mathrm{col}_{\beta_i\dots\beta_{j-1}}(z_1)))\in  e_{z_2}\\
     \iff\quad&\mathrm{col}_{\beta_i\dots\beta_{j-1}}(v_{z_1})\in \mathrm{eset}_\beta(z_2)\\
     \iff\quad&r^{j-i-1}(\mathrm{col}_{\beta_i}(z_1))\in \mathrm{eset}_\beta(z_2)\\
     \iff\quad&r^{n-i-1}(\mathrm{col}_{\beta_i}(z_1))\in \mathrm{eset}_\beta(z_2)\\
     \iff\quad&\mathrm{col}_{\beta_i\dots\beta_{n}}(z_1)\in \mathrm{eset}_\beta(z_2)\\
     \iff\quad&\mathrm{col}_{\beta}(z_1)\in \mathrm{eset}_\beta(z_2)\\\end{align*}
 \end{proof}

\begin{theorem}
  \label{thm:cog0}
  Let $f(k)=(k2^{k+1})^{3k^k}$ and $g(k)=3k^k$.  Every graph with
  linear NLC-width at most $k$ can be vertex partitioned into $f(k)$
  cographs with a cotree of depth at most $g(k)$.
\end{theorem}
\begin{proof}
  Let $\kappa$ be a coloring of the nodes $\beta$ with color in $[2]$
  such that two consecutive children of a node have a different color.
  For a letter $z$ of $\alpha$, color $v_z$ by the vector of values
  $(\kappa(\beta),\mathrm{col}_\beta(z),\mathrm{eset}_\beta(z))$ for
  $\beta$ ancestor of $z$. (This gives a vector of at most $3|T|$
  triples). Consider a monochromatic subset of vertices. It is easily
  checked that this set induces a cograph with cotree height at most
  $3|T|$.
\end{proof}
\begin{corollary}
  \label{cor:chiLNLC}
  Classes with bounded linear NLC-width are linearly $\chi$-bounded.
\end{corollary}

\begin{lemma}
  Assume there exists $\beta$ and letters
  $x_1,y_1,x_2,y_2,\dots,x_\ell,y_\ell$ of $\beta$ (in this order)
  such that $\beta$ is the least common ancestor of each pair of these
  letters, and that there exist $c_x,c_y\in [k]$ and
  $e_x,e_y\subseteq [k]$ with $c_x\in e_y, c_y\notin e_x$, and, for
  each $1\leq i\leq \ell$, $\mathrm{col}_\beta(x_i)=c_x$,
  $\mathrm{eset}_\beta(x_i)=e_x$, $\mathrm{col}_\beta(y_i)=c_y$, and
  $\mathrm{eset}_\beta(x_y)=e_y$.
  Then $\Xi(\alpha)$ contains a semi-induced half-graph of order at
  least $\lfloor\ell/3\rfloor$.
\end{lemma}
\begin{proof}
  By taking at least a third of the indices we can assume that no two
  letters appear in consecutive children of $\beta$. Then it follows
  directly from \Cref{lem:scode} that these vertices semi-induce
  a half-graph.
\end{proof}

\begin{theorem}
  Let $\Cc$ be a class with bounded linear NLC-width. If the graphs in
  $\Cc$ exclude some semi-induced half-graph, then $\Cc$ is a
  transduction of a class with bounded pathwidth.
\end{theorem}
\begin{proof}
  We first construct the interval graph $H$, where each node $\delta$
  of $Y$ corresponds to an interval $I_\delta$. The descendent
  relation of $Y$ is then the containment relation in the set of
  intervals.
	
  Now consider an internal node $\delta$ of $Y$ and a $4$-tuple
  $(c_1,e_1,c_2,e_2)\in [k]\times 2^{[k]}\times [k]\times 2^{[k]}$
  with $c_1\in e_2$ and $c_2\notin e_1$, such that at least one
  descendent~$z_1$ of~$\delta$ is such that
  $\mathrm{col}_\delta(z_1)=c_1$ and $\mathrm{col}_\delta(z_1)=e_1$
  and at least one descendent $z_2$ of $\delta$ is such that
  $\mathrm{col}_\delta(z_1)=c_2$ and $\mathrm{col}_\delta(z_1)=e_2$.
  We consider new intervals coming from the split of the $I_\delta$
  into subintervals: These subintervals are obtained by considering
  the children of $\delta$ in order. The subintervals are of three
  types:
  \begin{itemize}
  \item the type $(1)$ contain consecutive children with at least one
    children with $\mathrm{col}_\delta(z)=c_1$ and
    $\mathrm{col}_\delta(z)=e_1$, but no descendant $z$ with
    $\mathrm{col}_\delta(z)=c_2$ and $\mathrm{col}_\delta(z)=e_2$;
  \item the type $(2)$ contain consecutive children with at least one
    children with $\mathrm{col}_\delta(z)=c_2$ and
    $\mathrm{col}_\delta(z)=e_2$, but no descendant $z$ with
    $\mathrm{col}_\delta(z)=c_1$ and $\mathrm{col}_\delta(z)=e_1$;
  \item the type $(1+2)$ contains a single children with both a
    descendent $z_1$ with $\mathrm{col}_\delta(z_1)=c_1$ and
    $\mathrm{col}_\delta(z_1)=e_1$ and a descendent $z_2$ with
    $\mathrm{col}_\delta(z_1)=c_2$ and $\mathrm{col}_\delta(z_1)=e_2$.
  \end{itemize}
  The division of $I_\delta$ into subintervals is done in such a way
  that no two consecutive subintervals are both of type $(1)$ or both
  of type $(2)$. Note that such a division into subintervals, though
  not uniquely defined, always exists.
	
  Assume that the number of subintervals into which we divided
  $I_\delta$ is $N$. Then we can select, among the descendants of the
  distinct children of $\delta$ some vertices
  $\alpha_1,\beta_1,\dots,\alpha_n,\beta_n$ (with $n\geq N/4$) such
  that \mbox{$\mathrm{col}_\delta(\alpha_i)=c_1$},
  \mbox{$\mathrm{col}_\delta(\alpha_i)=e_1$},
  $\mathrm{col}_\delta(\beta_i)=c_2$, and
  $\mathrm{col}_\delta(\beta_i)=e_2$. It is easily checked that the
  vertices $\alpha_1,\beta_1,\dots$, $\alpha_n,\beta_n$ semi-induce a
  half-graph of order $n$. As $\Cc$ excludes some semi-induced
  half-graph we deduce that $I_\delta$ is divided into a bounded
  number of subintervals, which can be numbered using a bounded number
  of unary predicates.

  Let $u,v$ be vertices, and let $\delta$ be their least common
  ancestor in $Y$. The values of $\mathrm{col}_\delta$ and
  $\mathrm{eset}_\delta$ for $u$ and $v$ are known from the predicates
  at these vertices.  Let $c_1=\mathrm{col}_\delta(u)$,
  $e_1=\mathrm{eset}_\delta(u)$, $c_2=\mathrm{col}_\delta(v)$, and
  $e_2=\mathrm{eset}_\delta(v)$. If $c_1\in e_2$ and $c_2\in e_1$ then
  $u$ and $v$ are adjacent. If $c_1\notin e_2$ and $c_2\notin e_1$
  then $u$ and $v$ are non-adjacent. In the last case, without loss of
  generality, we can assume $c_1\in e_2$ and $c_2\notin e_1$.
	
  The two vertices $u$ and $v$ cannot belong to a same subinterval of
  $I_\delta$. From the numbering marks associated to the subintervals
  that contain $u$ and $v$ we deduce which of $u$ and $v$ is smaller
  than the other and hence the adjacency between~$u$ and $v$.
\end{proof}

From this we deduce.
\begin{theorem}
  \label{thm:stlrw}
  Let $\Cc$ be a class of graphs with linear rankwidth at most
  $r$. Then the following are equivalent:
  \begin{enumerate}
  \item $\Cc$ is stable,
  \item $\Cc$ is monadically stable,
  \item $\Cc$ is sparsifiable,
  \item $\Cc$ has $2$-covers with bounded shrubdepth,
  \item $\Cc$ has structurally bounded expansion,
  \item $\Cc$ is a transduction of a class with bounded pathwidth,
  \item $\Cc$ excludes some semi-induced half-graph.
  \end{enumerate}
\end{theorem}

\section{Linear rankwidth}
\label{sec:lrw}
\smallskip\noindent In this section we present a second
proof for the result that classes with bounded linear rankwidth
are linearly $\chi$-bounded and thereby provide improved
constants. 

\subsection{Notation}

\smallskip\noindent
For sets $M,N\subseteq V(G)$ we define $M\oplus N$ as 
the symmetric difference of $M$ and $N$, that is, 
$v\in M\oplus N$ if and only if $v\in M\cup N$ but 
$v\notin M\cap N$. 
%
%
For $t\in V$, we define $V^{>t}\coloneqq 
\{v : v>t\}$, $V^{<t}\coloneqq 
\{v : v<t\}$ and $V^{\leq t}\coloneqq \{v : v\leq t\}$. 
For $v\in V$ we denote by~$N(v)$ the neighborhood 
of $v\in G$ (where $v$ not included). We let $N^{<t}(v)\coloneqq N(v)\cap V^{<t}$ and define similarly $N^{>t}$ and $N^{\leq t}$.
For $M\subseteq V(G)$ we define $N_\oplus(M)\coloneqq
\bigoplus_{v\in M} N(v)$ and $N_\oplus^{>t}(M)\coloneqq
N_\oplus(M)\cap V^{>t}$. 

\begin{remark}\label{rem:larger-neighborhoods}
If $t<t'$, then $N_\oplus^{>t}(M)=N_\oplus^{>t}(N)$ implies 
$N_\oplus^{>t'}(M)=N_\oplus^{>t'}(N)$. 
\end{remark}

For $t\in V$ the closure of $\{N^{>t}(v) : v\leq t\}$ under
$\oplus$ is a vector space over $\oplus$ and scalar multiplication
with $0$ and~$1$, where $0\cdot M=\emptyset$ and 
$1\cdot M=M$. 

For $t\in V$, we call an inclusion-minimal subset $B\subseteq V_{\leq t}$ a \emph{neighbor basis 
for $V^{>t}$} if for every
$v\leq t$ there exists $B'\subseteq B$ such that 
$N^{>t}(v)=N_\oplus^{>t}(B')$. 
In other words, $B$ is a neighbor basis for $V^{>t}$ if 
$\{N^{>t}(v)\mid v\in B\}$ forms a basis for the space spanned by $\{N^{>t}(v)\mid v\leq t\}$. 

\smallskip
The following is 
immediate by the definition of linear rankwidth. 

\begin{remark}
As $G$ has linear rankwidth at most $r$, for every $t\in V$ 
every neighbor basis for $V^{>t}$ of order at most $r$. 
\end{remark}

\subsection{Activity intervals and active basis}

\smallskip\noindent
For
$t\in V$ we define the \emph{active basis $B_t$ at~$t$} as the set
\begin{equation}
B_t=\{v\leq t\mid (\nexists B\subseteq V^{<v})\ N^{>t}(v)=N_\oplus^{>t}(B)\}.	
\end{equation}

Note that this is  the
lexicographically least neighborhood basis of $V^{>t}$. 
\begin{remark}
	If the linear order of $V(G)$ is given, the set of all neighborhood basis $B_t$ for $t\in V(G)$ can be computed in linear time.
\end{remark}
%

To each $v\in V$ we associate its {\em activity
  interval} $I_{v}$ defined as the interval $[v,\tau(v)] $
starting at $v$ and ending at the minimum vertex $\tau(v)\geq v$
such that $v\notin B_{\tau(v)}$. 
%
Note that $\tau(v)$
is well defined as we have $B_{\max V}=\emptyset$.



We extend the definitions of the activity intervals and of the $\tau$ function to all subsets $M$ of $V(G)$ by
\begin{equation}
	I_M\coloneqq \bigcap_{v\in M}I_v\quad \text{ and } \quad \tau(M)=\min_{v\in M}\tau(v).
\end{equation}

Note that either $I_M=\emptyset$ or $I_M=[\max M,\tau(M)]$.
We call a set $M$ {\em active} if $|I_M|>1$, that is, if
$\max M<\tau(M)$. 
We call a vertex $v$ \emph{active} if the singleton set $\{v\}$
is active. 


\smallskip
For every $v\in V$, as $v\notin B_{\tau(v)}$, there 
exists a unique $F_0(v)\subseteq B_{\tau(v)}$ with 
\begin{equation}
\label{eq:F0}
N^{>\tau(v)}(v) = N_\oplus^{>\tau(v)}(F_0(v)).
\end{equation}	

Note that if $F_0(v)\neq \emptyset$, then we have 
\begin{equation}
\label{eq:interval}
\max F_0(v)<v\leq \tau(v)<\tau(F_0(v)).
\end{equation}

Hence, in this case, the set $F_0(v)$ is active.

\begin{remark}\label{rem:active-F}
Assume that $M$ is an active set and let $v\in M$. 
\begin{enumerate}
\item If 
$\tau(v)> \tau(M)$, then $v\in B_{\tau(M)}$. 
\item If $\tau(v)= \tau(M)$, then $F_0(v)
\subseteq B_{\tau(M)}$. 
\end{enumerate}
\end{remark}



\subsection{The F-tree}

\smallskip\noindent
We define a mapping $F$ extending $F_0$, that will define a rooted tree on the set  $Z$ consisting of all active sets, all singleton sets $\{v\}$
for $v\in V(G)$, and $\emptyset$ (which will be the root of the tree and the unique fixed point of $F$). Before we define $F$ we make 
one more observation. 

\begin{lemma}
\label{lem:tau}
	Let $u,v\in V(G)$ be active. If $\tau(u)=\tau(v)$, then $u=v$.
\end{lemma}
\begin{proof}
  Let $t=\tau(u)=\tau(v)$ and let $t'$ be the predecessor of $t$ in the linear order. Assume for contradiction that $u\neq v$.   By definition of $F_0$ we have
	$N^{>t}(u)=N^{>t}(F_0(u))$ and $N^{>t}(v)=N^{>t}(F_0(v))$.
	We have $N^{>t'}(u)\neq N^{>t'}(F_0(u))$ as otherwise
	$\tau(u)\leq t'$. As $N^{>t'}(u)\oplus N^{t}(u)\subseteq \{t\}$ and $N^{>t'}(F_0(u))\oplus N^{t}(F_0(u))\subseteq \{t\}$, we have
	$N^{>t'}(F_0(u))=N^{>t'}(u)\oplus\{t\}$.
	Similarly, we have $N^{>t'}(F_0(v))=N^{>t'}(v)\oplus\{t\}$. Assume without loss of generality that $u<v$. Then
	$N^{>t'}(v)=N^{>t'}(\{u\})\oplus N^{>t'}(F_0(u))\oplus N^{>t'}(F_0(v))$. 
	As $\max (\{u\}\cup F_0(u)\cup F_0(v))<v$ we deduce that $\tau(v)\leq t'$, contradicting $\tau(v)=t$.
\end{proof}

\begin{corollary}
For each active set $M\subseteq V(G)$ there exists exactly one $v\in M$
with $\tau(v)=\tau(M)$. 
\end{corollary}

The mapping $F\colon Z\rightarrow Z$ is defined as 
\begin{equation}
	F(M)=\begin{cases}
		\emptyset&\text{if }M=\emptyset,\\
M\oplus \{v\}\oplus F_0(v) &\text{for the 
unique $v\in M$}\\
& \text{with $\tau(v)=\tau(M)$, otherwise}.
	\end{cases}
\end{equation}
\begin{remark}
	If the linear order on $V(G)$ is given then $F$-mapping on $Z$ can be computed in linear time. (Note that $|Z|\leq 2^r|V(G)|$.)
\end{remark}
\medskip

The following lemma shows for every active set $M$, either 
$F(M)=\emptyset$ or $F(M)$
is active, and thus $F(M)\in Z$ and~$F$ is well defined. 
Furthermore, the lemma shows that 
$I_{F(M)}\supset I_M$. 

\begin{lemma}
\label{claim:tauinc}
Let $M\in Z$. Then $F(M)\subseteq B_{\tau(M)}$ and furthermore, 
either $F(M)=\emptyset$, or 
$\max F(M)\leq \max M<\tau(M)<\tau(F(M))$ and 
hence $F(M)$ is active. 
\end{lemma}
\begin{proof}
The statement is obvious if $M=\emptyset$. 
For $M=\{v\}$, the statement is immediate from the 
definition of $F_0(v)$ and~\cref{eq:interval}. 
For all other $M\in Z$, according to \cref{rem:active-F}
we have for each $v\in M$
either $v\in B_{\tau(M)}$ if $\tau(v)> \tau(M)$, 
or $F_0(v)\subseteq B_{\tau(M)}$ if $\tau(v)=\tau(M)$. 
This implies $F(M)\subseteq B_{\tau(M)}$. 
Finally, if $F(M)\neq \emptyset$, then 
$\max F(M)\leq \max M<\tau(M)<\tau(F(M))$ follows
from the fact that these inequalities hold for all 
$v\in M$ with $\tau(v)>\tau(M)$ and for $F_0(v)$
for the unique $v\in M$ with $\tau(v)=\tau(M)$ according to \cref{eq:interval}. 
\end{proof}

\smallskip
The mapping $F$ guides the process of iterative referencing
and ensures that, for an active set~$M$, if $t\geq\tau(M)$, then the 
set $N_\oplus^{>t}(M)$ can be rewritten as 
$N_\oplus^{>t}(F(M))$. 
This property is stated in the next lemma.

\begin{lemma}
\label{lem:F}
Let  $M\in Z\setminus\{\emptyset\}$ and let $w\in V(G)$. 
If $w>\tau(M)$, then 
\[
    w\in N_\oplus(M) \Leftrightarrow w\in N_\oplus(F(M)). 
\]
\end{lemma}
\begin{proof}
If $M=\{v\}$ for $v\in V(G)$, then this follows from~\eqref{eq:F0}.
Otherwise, $M$ is an active set.
Let $t=\tau(M)$ and let $v\in M$ be the unique element
with $\tau(v)=t$. 
Then we have $N_\oplus^{>t}(F(v))=
N_\oplus^{>t}(v)$, and hence 
\begin{align*}
N_\oplus^{>t}(F(M))
&= N_\oplus^{>t}(\{v\})\oplus N_\oplus^{>t}(F(M)\oplus\{v\})\\
&= N_\oplus^{>t}(F_0(v))\oplus N_\oplus^{>t}(F(M)\oplus\{v\})\\
&=N_\oplus^{>t}(M).	
\end{align*}
%
\end{proof}

This lemma can be applied repeatedly to $M, F(M),$ etc.\ until 
$F^k(M)=\emptyset$, or until for some given $w\in V(G)$ we
have $\tau(F^k(M))\geq w$. This justifies to introduce, for 
distinct vertices $u$ and $v$ the value
\begin{equation}
	\xi(u,v):=\min\{k\mid v\in I_{F^k(u)}\text{ or }F^k(u)=\emptyset\}
\end{equation}



\smallskip
As a direct consequence of the previous lemma we have
\begin{corollary}
\label{cor:fund}
For distinct $u,v\in V(G)$ we have
\[
\{u,v\}\in E(G) \Longleftrightarrow
\begin{cases}
	v\in N_\oplus(F^{\xi(u,v)}(u))&\text{if }u<v,\\
	u\in N_\oplus(F^{\xi(v,u)}(v))&\text{if }u>v.
\end{cases}
\]
\end{corollary}
\begin{proof}
As the two cases are symmetric, we can assume $u<v$. 
If $k=0$, then the statement is 
$\{u,v\}\in E(G)\Leftrightarrow v\in N_\oplus(u)$, which 
trivially holds. 
Assume $\xi(u,v)=k\geq 1$. By \Cref{claim:tauinc} we have 
	$v>\tau(F^{k-1}(\{u\}))>\tau(F^{k-2}(\{u\}))>\dots>\tau(u)$.
	Moreover, $u,F(\{u\}),\dots, F^{k-1}(\{u\})\in 
	Z\setminus\{\emptyset\}$. 
	Hence  by \Cref{lem:F} we have
\begin{align*}
\{u,v\}\in E(G) & \Leftrightarrow v\in N_\oplus(u) \\
& \Leftrightarrow v\in N_\oplus(F(\{u\}))\\
& \Leftrightarrow v\in N_\oplus(F^2(\{u\})) \Leftrightarrow\ldots\\
& \Leftrightarrow v\in N_\oplus(F^k(\{u\})).
\end{align*}
\end{proof}

The monotonicity property of $F$ (i.e. the property $\tau(F(M))>\tau(M)$ if $F(M)\neq \emptyset$) implies that $F$ defines a rooted tree, \emph{the $F$-tree}, with vertex set~$Z$, root $\emptyset$ and edges $\{M, F(M)\}$. Here the monotonicity guarantees that the graph is acyclic and it is connected because~$\emptyset$ is the only fixed point of~$F$. The following lemma shows that the $F$-tree has bounded height. Recall that~$r$ denotes the linear rankwidth of~$G$. 

\begin{lemma}
	For every $M\in Z$ we have $F^{r+1}(M)=\emptyset$.
\end{lemma}
\begin{proof}
If $M=\emptyset$, the statement is obvious, so assume $M\neq \emptyset$. It is sufficient 
to prove that for every active set $M$ we have 
$F^{r}(M)=\emptyset$, as this implies $F^{r+1}(\{v\})=\emptyset$
also for all $v\in V(G)$.
Let $M$ be an active set and let $t \in I_M$. Then every $v\in M$ is in $B_t$, so $M\subseteq B_t$. 

Assume $i\geq 1$ is such that $F^i(M)\neq \emptyset$. 
As $\max F(M)\leq \max M$ and $\tau(F(M))>\tau(M)$ 
by \Cref{claim:tauinc}, we get
\[
\max F^i(M)\leq \max M\leq t<\tau(M)\leq \tau(F^{i-1}(M))<\tau(F^i(M)).
\]
As $\tau(F^i(M))=\min_{v\in F^i(M)}\tau(v)$, we 
have $F^i(M)\subseteq B_t$. Hence, 
considering the sequence $M, F(M), \dots, F^i(M)$, each 
iteration of~$F$ removes the unique element with minimum 
$\tau$ value. It follows that the union of the sets 
has cardinality at least $i+1$. As $|B_t|\leq r$, we have $i<r$ and
hence $F^r(M)=\emptyset$.
\end{proof}

For distinct vertices $u,v$, let $u\wedge v$ denote the greatest common ancestor of $u$ and $v$ in the $F$-tree, i.e. the first common vertex on the paths to the root.
Then there exist $\ell_u$ and~$\ell_v$ such that
$u\wedge v=F^{\ell_u}(u)=F^{\ell_v}(v)$, hence both $u$ and $v$ belong 
to $I_{u\wedge v}$. Thus we have $\tau(u\wedge v)>u$ and $\tau(u\wedge v)>v$. In other words, we have
 $\xi(u,v)\leq \ell_u$ and $\xi(v,u)\leq \ell_v$.


\subsection{The activity interval graph}

\smallskip\noindent
Let $H$ be the intersection graph of the intervals $I_v$ 
for $v\in V(G)$. Note that we may identify $V(H)$ with 
$V(G)$ as $\min I_v=v$ for all $v\ V(G)$.

\begin{lemma}
The intersection graph $H$ of the intervals $I_u$ has pathwidth
at most $r+1$, i.e. at most $r+2$ intervals intersect in each point. 
\end{lemma}
\begin{proof}
Consider any vertex $t$ with $t\in I_u$ for some $u$. The case $u\in B_t$ 
gives a maximum of $r$ intervals intersecting in~$t$. Otherwise 
$t=\tau(u)$, which gives at most two possibilities for $u$: 
either $u$ is inactive (and $u=t$), or $u$ is active
(and $u$ is uniquely determined, according to \cref{lem:tau}).
Thus at most $r+2$ intervals intersect at point~$t$.
\end{proof}

As mentioned in the proof of the above lemma, every clique of $H$ contains at most one inactive vertex.
It follows that there is a coloring $\gamma\colon V(G)\rightarrow[r+2]$ with the following properties:
\begin{enumerate}[(1)]
	\item for every $u\in V(G)$ we have $\gamma(u)=r+2$ if and only if $u$ is inactive;
	\item for all distinct $u,v\in V(G)$ we have
\begin{equation}
I_u\cap I_v\neq\emptyset\quad\Longrightarrow\quad\gamma(u)\neq\gamma(v).
\end{equation}
\end{enumerate} 

We extend this coloring to sets as follows:
for $M\subseteq V(G)$ we let
\begin{equation}
\Gamma(M)\coloneqq \{\gamma(v)\mid v\in M\}.	
\end{equation}

This coloring allows to define, for each $v\in V(G)$ 
\begin{align*}
	\Class(v)&\coloneqq \big(\gamma(v),\Gamma(F(v)),\dots,\Gamma(F^r(v))\big),\\
	\NC(v)&\coloneqq \{\gamma(u)\mid u\in N(v)\text{ and }v\in I_u\}
\end{align*}

Note that all $u$ with $v\in I_u$ define a 
clique of $H$ (because all $I_u$ contain~$v$) and hence have distinct $\gamma$-colors.

\begin{lemma}
\label{cl:gamma}
	Let $v\in V(G)$. Every $u\in B_v$ can be defined as the 
	maximum vertex $x\leq v$ with $\gamma(x)=\gamma(u)$. 
\end{lemma}
\begin{proof}
	By assumption we have $u\leq v$. Assume towards a contradiction that there exists $x\in V(G)$ with $u<x\leq v$ and $\gamma(x)=\gamma(u)$. As $u\in B_v$ we have $\tau(u)>v$, hence $x\in I_u$. It follows that 
	$I_x\cap I_u\neq \emptyset$, in contradiction to $\gamma(x)=\gamma(u)$.
\end{proof}


Towards the aim of bounding the number of graphs of linear 
rankwidth at most $r$, we give a bound on the number of 
colors that can appear. 

\begin{lemma}
\label{cl:f}
Let $f(r)\coloneqq 3(r+2)!\,2^{\binom{r+1}{2}}$. 
The number of pairs $(\Class(v),\NC(v))$ for $v\in V(G)$ can be bounded by~$f(r)$.
\end{lemma}
\begin{proof}
		Let $v\in V(G)$.
	From the fact that $\gamma(v)=r+2$ if and only if $v$ is inactive, that images by $F$ only contain active vertices, as well as from \cref{claim:tauinc} we deduce:
\begin{itemize}
	\item 	If $\gamma(v)=r+2$, then there exists a linear order on $[r+1]$ colors such that for $1\leq i\leq r$, the set 
	$\Gamma(F^i(v))$ is a subset of the first $r+1-i$ colors of $[r+1]$.
	\item  If $\gamma(v)\leq r+1$, then there exists a linear order on $[r+1]\setminus\{\gamma(v)\}$ such that for $1\leq i\leq r$, the set $\Gamma(F^i(v))$ is a subset of the first $r-i$ colors of $[r]$.
\end{itemize}

Thus the number of distinct $\Class(v)$ for $v\in V(G)$ is bounded by
\[
(r+1)!\,2^{r}2^{r-1}\dots 2+(r+1)r!\,2^{r-1}\dots 2=3(r+1)!\,2^{\binom{r}{2}}.
\]
Furthermore, the number of distinct $\NC(v)$ for $v\in V(G)$ is at most $(r+2)2^{r+1}$.
\end{proof}


\subsection{Encoding the graph in the linear order}\label{subsec:encoding}

\smallskip\noindent
We first make use of \Cref{cor:fund} to encode $G$ 
by a first-order formula using only the newly added colors
and the order~$<$ on $V(G)$. More precisely, 
for $v\in V(G)$, 
let 
\[\IC(v)\coloneqq \{\gamma(u)\mid v\in I_u\}.\]
Let $\mathcal{L}$ be the structure over signature 
$\Lambda\mathbin{\cup} \{<\}$, where $\Lambda$ is the
set of all colors of the form $(\Class(v), 
\NC(v),\IC(v))$, with the same elements as $G$ and $<$ interpreted as
in $G$. Every element $v$ of $\mathcal{L}$ is equipped with 
the color $(\Class(v),$ $\NC(v),\IC(v))$. The following lemma gives a new
proof of the result of \cite{colcombet2007combinatorial}.

\begin{lemma}
There exists an $\exists\forall$-first-order formula $\phi(x,y)$ over the 
vocabulary $\Lambda\mathbin{\cup} \{<\}$ such that
for all $u,v\in V(G)$ we have 
\[\mathcal{L}\models\phi(u,v)\Longleftrightarrow \{u,v\}\in E(G).\]
\end{lemma}
\begin{proof}
By symmetry, we can assume that $u<v$. According to \cref{cor:fund} for distinct $u,v\in V(G)$ we have
\[
\{u,v\}\in E(G) \Longleftrightarrow
\begin{cases}
	v\in N_\oplus(F^{\xi(u,v)}(u))&\text{if }u<v\\
	u\in N_\oplus(F^{\xi(v,u)}(v))&\text{if }u>v.
\end{cases}
\]

Note that we can extract any color from $\Lambda$, i.e.\ we can define $\gamma(x) \in \Gamma(F^i(y))$ and $\gamma(x) \in \IC(y)$. For example, $\gamma(x) \in \Gamma(F^i(y))$ is a big disjunction over all possible colorings 
$\Lambda(x) = (\Class(x),NC(x),\IC(x))$ and $\Lambda(y) = (\Class(y),NC(y),\IC(y))$ satisfying that $\Class(x)$ has in 
its first component an element from the $i$th component of $\Class(y)$.

We first define formulas $\psi^i(x,y)$ such that for 
all $u,v\in V(G)$
\[\mathcal{G}\models \psi^i(u,v)\Leftrightarrow v\in F^i(u).\]

Let $C=\Gamma(F^i(u))$. According to \cref{cl:gamma}, 
for $a\in C$, the element of $F^i(u)\subseteq B_u$ 
with color~$a$ is the maximal element 
$w<u$ such that $\gamma(w)=a$. The formula can express
that $y< x$ is maximal with $\gamma(y)=a$ by 
$(y<x)\wedge (\gamma(y)=a)$ $\wedge \forall z\, ((z>y)\wedge(z<x) \rightarrow \gamma(z)\neq a)$. 
Here, for convenience, 
we use $\gamma(z)=a$ as an atom. Note that $\psi^i(x,y)$ is a 
$\forall$-formula. 

\smallskip
We now define formulas $\alpha^k(x,y)$ such that for all 
$u,v\in V(G)$
with $u<v$ we have 
\[\mathcal{G}\models\alpha^k(u,v)\Leftrightarrow k=\xi(u,v).\]

Observe that $v\in I_{F^k(u)}$ if and only if for every $x\in F^k(u)$ we have $x\leq v$, $a\in \IC(v)$ (i.e. there exists some $y$ with $\gamma(y)=a$ and $v\in I_y$) and there exists no $z$ with $x<z\leq v$ with $\gamma(z)=a$ (hence $\min I_y\leq x$, which implies that~$I_y$ and~$I_x$ intersects thus $x=y$ as $\gamma(x)=\gamma(y)$).
We restrict ourselves to the case $u<v$ and obtain 
\begin{equation*}
\begin{split}
u<v \land v\in I_{F^k(u)}\iff 
  u<v  \land{} &\Gamma(F^k(u))\subseteq \IC(v)\\ 
  {}\land &\forall x\, (x\in F^k(u) \limp x \leq v \land \gamma(x)\notin \IC(v))\,.
\end{split}
\end{equation*}

Then $\xi(u,v)$ for $u<v$ is the minimum integer~$k$ such that $v\in I_{F^k(u)}$ or $F^k(u)=\emptyset$, and
  this is easy to state as a $\forall$-formula. 
\smallskip
Finally, if we have determined $\xi(u,v)$, with the help of the
formulas $\psi^i$ we can determine whether $\{u,v\}\in E(G)$
as in the proof of \cref{cor:fund} by existentially quantifying
the elements of $F(u), F^2(u),\ldots, F^{\xi(u,v)}(u)$ and 
expressing whether $v\in N_\oplus(F^{\xi(u,v)}(u))$. 
Indeed, for every $x\in F^{\xi(u,v)}(u)$ we have $v\in I_{F^{\xi(u,v)}(u)}\subseteq I_x$, hence the adjacency of $x$ and $y$ is encoded in $\NC(v)$.

This information can hence
be retrieved by an $\exists\forall$-formula, as claimed. 
\end{proof}

%

\begin{lemma}
	Let $f'(r)\coloneqq (r+2)!\,2^{\binom{r}{2}}3^{r+2}$.
	 The number of triples $(\Class(v),\NC(v),\IC(v))$ for $v\in V(G)$  can be bounded by $f'(r)$.
\end{lemma}
\begin{proof}
	In \Cref{cl:f} we have shown that the number of distinct $\Class(v)$ for $v\in V(G)$ is bounded by $3(r+1)!\,2^{\binom{r}{2}}$.
	The number of pairs $(\NC(v),\IC(v))$ is at most $(r+2)3^{r+1}$ (for each color $a$ in $[r+1]$ either
	$a\notin\IC(v)$ or $a\in\IC(v)\setminus \NC(v)$ or $a\in \NC(v)$).
\end{proof}

As a corollary we conclude an upper bound on the number of 
graphs of bounded linear rankwidth. 
\begin{theorem}
\label{thm:numlrw}
	Unlabeled graphs with linear rankwidth at most $r$ can be encoded using at most 
	$\binom{r}{2}+r\log_2 r+\log_2(3/e)r+O(\log_2 r)$ bits per vertex.
	Precisely, the number of unlabelled graphs of order $n$ with linear rankwidth at most $r$ is at most $\left[(r+2)!\,2^{\binom{r}{2}}3^{r+2}\right]^{n}$.
\end{theorem}

\begin{remark}
	The encoding can be computed in linear time if the linear order on $G$ is given.
\end{remark}


\subsection{Partition into cographs}
\label{sec:cograph}


\begin{theorem}
\label{thm:cog}
Let $f(r)=3(r+2)!\,2^{\binom{r+1}{2}}$.
The $c$-chromatic number of every graph $G$ is bounded by $f(\lrw(G))$ and hence 
\begin{equation}
	\chi(G)\leq f(\lrw(G))\,\omega(G).
\end{equation}
\end{theorem}
\begin{proof}
	Let $u\sim v$ hold if and only if $\Class(u)=\Class(v)$ and $\NC(u)=\NC(v)$. 
As proved in \Cref{cl:f} there are at most $f(r)$ equivalence classes for the relation $\sim$.
		
	Let $X$ be an equivalence class for $\sim$, and let $u,v$ be distinct elements in $X$.
	Let $k=\xi(u,v)$ and let $\ell=\xi(v,u)$. 

	If $F^k(u)=\emptyset$, then $F^k(v)=\emptyset$ as $\Class(v)=\Class(u)$. 
	Otherwise, $F^k(u)\neq \emptyset$, thus $F^k(v)\neq \emptyset$. As $v\in I_{F^k(u)}$ and $v\in I_{F^k(v)}$ we deduce that 
	$F^k(u)$ and $F^k(v)$ are both included in $B_v$. As the 
	vertices of a given color in $B_v$ are uniquely determined 
	we deduce $F^k(u)=F^k(v)$. 
	Similarly, we argue that $F^\ell(u)=F^\ell(v)$. It follows that $F^k(u)=F^\ell(u)=u\wedge v$.
	
Hence, if $x\wedge y=u\wedge v$ for $x,y\in X$, then we have
$x\wedge y=F^k(x)=F^k(u)$. 
As $\NC(u)=\NC(v)$, we deduce that for all $x,y\in X$ with 
$x\wedge y=u\wedge v$ we have 
$y\in N_\oplus(F^k(x))$ or for all $x,y\in X$ 
with $x\wedge y=u\wedge v$ we have 
$y\not\in N_\oplus(F^k(x))$. Then it follows from \cref{cor:fund}
that at each inner vertex
of $F$ on $X$ we either define a join or a union. Hence, 
$G[X]$ is a cograph with cotree $F$ restricted to 
$X$ of height at most $r+2$. 
\end{proof}

\begin{remark}
	The partition can be computed in linear time if the ordering of the vertex set is given.
\end{remark}

The function $f(r)$ is most probably far from being optimal. This naturally leads to the following question.

\begin{problem}
	Estimate the growth rate of function $g:\mathbb N\rightarrow\mathbb R$ defined by
\begin{equation}
	g(r)=\sup\,\biggl\{\frac{\chi(G)}{\omega(G)}\mid\lrw(G)\leq r\biggr\}\,.
\end{equation}
\end{problem}
\begin{remark}
One may wonder whether bounding $\chi(G)$ by an affine function of $\omega(G)$ could decrease the coefficient of $\omega(G)$. In other words, is the ratio  $\chi/\omega$ be asymptotically much smaller (as $\omega\rightarrow\infty$) than its supremum?
Note that if $\lrw(G)=r$ and $n\in\mathbb N$, then the graph $G_n$ obtained as the join of $n$ copies of $G$ satisfies $\lrw(G_n)\leq r+1$, $\omega(G_n)=n\omega(G)$ and $\chi(G_n)=n\chi(G)$. Thus

\[
g(r-1)\leq \limsup_{\omega\rightarrow\infty}\,\biggl\{\frac{\chi(G)}{\omega(G)}\biggr|\lrw(G)\leq r\text{ and }\omega(G)\geq\omega\biggr\}\leq g(r).
\]	
\end{remark}

\begin{problem}
	Is the ratio $\chi(G)/\omega(G)$ bounded by a polynomial function of the neighborhood-width of $G$ (equivalently, of the linear cliquewidth or of the linear NLC-width of $G$)?
\end{problem}

\section{Conclusion, further works, and open problems}

\smallskip\noindent
In this paper, several aspects of classes with bounded linear-rankwidth have been studied, both from  (structural) graph theoretical and the model theoretical points of view.

On the one hand, it appeared that graphs with bounded linear rankwidth do not form a ``prime'' class, in the sense that they can be further decomposed/covered using pieces in classes with bounded embedded shrubdepth. As an immediate corollary we obtained that classes with bounded linear rankwidth are linearly $\chi$-bounded.
Of course, the $\chi/\omega$ bound obtained in \Cref{thm:cog} is most probably very far from being optimal.

On the other hand, considering how graphs with linear rank-width at most $r$ are encoded in a linear order or in a graph with bounded pathwidth with marginal ``quantifier-free'' use of a compatible linear order improved our understanding of this class in the first-order transduction framework. 

%


Classes with bounded rankwidth seem to be much more complex than expected and no simple extension of the results obtained from classes with bounded linear rankwidth seems to hold. In particular, these classes seem to be ``prime'' in the sense that you cannot even partition the vertex set into a bounded number of parts, each inducing a graph is a simple hereditary class like the class of cographs (see \Cref{cor:herw}). 
However, the following conjecture seems reasonable to us.

\begin{conjecture}\label{conjecture:rw}
  Let $\Cc$ be a class of graphs of bounded rankwidth. Then $\Cc$ has structurally  bounded treewidth if
  and only if $\Cc$ is stable.
\end{conjecture}


We believe that our study of classes with bounded linear rankwidth might  open the perspective to study classes admitting low linear rankwidth covers. Let us elaborate on this.
As a consequence of \Cref{thm:stlrw} we have the following:
\begin{theorem}
\label{thm:lowlrwcover}
	Let $\Cc$ be a class with low linear rankwidth covers. Then the following are equivalent:
	\begin{enumerate}
		\item\label{it:lrwc1} $\Cc$ is monadically stable,
		\item\label{it:lrwc2} $\Cc$ is stable,
		\item\label{it:lrwc3} $\Cc$ excludes a semi-induced half-graph,
		\item\label{it:lrwc4} $\Cc$ has structurally bounded expansion.
	\end{enumerate}
\end{theorem}
\begin{proof}
	Clearly $\ref{it:lrwc1}\Rightarrow\ref{it:lrwc2}\Rightarrow\ref{it:lrwc3}$.	
For $\ref{it:lrwc3}\Rightarrow\ref{it:lrwc4}$, let $p$
be an integer and consider a depth-$p$ cover $\cal U$ of $G\in\Cc$ with linear rankwidth at most $r$. If $\Cc$ excludes some semi-induced half-graph we deduce by \Cref{thm:stlrw} that each $U\in \cal U$ 
induces a subgraph that is a fixed transduction of a graph with pathwidth at most $C(r)$, hence, of a class
that has depth-$p$ covers with bounded shrubdepth. Considering the intersection of the two covers, we get that $\Cc$ has depth-$p$ covers with bounded shrubdepth, hence, has structurally bounded expansion. Thus $\ref{it:lrwc3}\Rightarrow\ref{it:lrwc4}$. Finally, $\ref{it:lrwc4}\Rightarrow\ref{it:lrwc1}$ is implied by 
\Cref{thm:Adler}.
\end{proof}

The next example illustrates again the concept of simple 
transductions and as a side product will provide us with 
some examples of classes of graphs admitting low 
linear rankwidth covers.

\begin{example}
\label{ex:Lozin}
We consider the following graph classes, introduced in \cite{lozin2011minimal}.
Let $n,m$ be integers. The graph $H_{n,m}$ has vertex set 
$V=\{v_{i,j}\mid (i,j)\in [n]\times [m]\}$. In this graph, two vertices $v_{i,j}$ and $v_{i',j'}$ with $i\leq i'$ are adjacent if $i'=i+1$ and $j'\leq j$. The graph $\widetilde{H}_{n,m}$ is obtained from $H_{n,m}$ by adding all the edges between vertices having same first index (that is between $v_{i,j}$ and $v_{i,j'}$ for every $i\in [n]$ and all distinct $j,j'\in [m]$.

First note that for fixed $a\in\mathbb N$ the classes 
$\mathscr H_a=\{H_{a,m}\mid m\in\mathbb N\}$ and 
$\widetilde{\mathscr H}_a=\{\widetilde{H}_{a,m}\mid m\in\mathbb N\}$ have bounded linear rank-width as they can be obtained as interpretations of  $a$-colored linear orders:
we consider the linear order on $\{v_{i,j}\mid (i,j)\in [a]\times [m]\}$ defined by $v_{i,j}<v_{i',j'}$ if $j<j'$ or $(j=j')$ and $(i<i')$. We color $v_{i,j}$ by color $i$. Then the graphs in $\mathscr H_a$ are obtained by the interpretation stating that $x<y$ are adjacent if the color of $x$ is one less than the color of $y$, and if there is no $z$ between $x$ and $y$ with the same color as $x$. The graphs  in $\widetilde{\mathscr H}_a$ are obtained by further adding all the edges between vertices with same color.
\end{example}

Following the lines of \cite[Theorem 9]{kwon17} we deduce from \Cref{ex:Lozin}:
\begin{proposition}
	The class of unit interval graphs and the class of bipartite permutation graphs admit
low linear rank-width colorings.
\end{proposition}

As we have shown above, classes with low linear rankwidth covers generalize structurally bounded expansion classes. Among the first problems to be solved on these class, two arise very naturally:

\begin{problem}
\label{pb:lrwT}
Is it true that every first-order transduction of a class with low linear rankwidth covers has again low linear rankwidth covers?	
\end{problem}

As a stronger form of this problem, one can also wonder whether classes with low linear rankwidth covers enjoy a form of quantifier elimination, as structurally bounded expansion class do.

\begin{problem}
\label{pb:lrwNIP}
Is it true that every class with low linear-rankwidth covers is mondadically NIP?	
\end{problem}
Note that it is easily checked that a positive answer to \Cref{pb:lrwT} would imply a positive answer to \Cref{pb:lrwNIP}.

\bibliographystyle{elsarticle-harv}
\bibliography{stablerw}

\end{document}